\documentclass[11pt,letterpaper]{article}

\usepackage{typearea}
\typearea{15}

\usepackage{latexsym,graphicx}
\usepackage{amsmath,amssymb,enumerate}
\usepackage{amsthm}
\usepackage{thmtools, thm-restate}

\newtheorem{theorem}{Theorem}[section]
\newtheorem*{theorem*}{Theorem}

\newtheorem*{proposition*}{Proposition}
\newtheorem{lemma}[theorem]{Lemma}
\newtheorem*{lemma*}{Lemma}
\newtheorem{corollary}[theorem]{Corollary}
\newtheorem*{conjecture*}{Conjecture}

\newtheorem*{fact*}{Fact}

\newtheorem*{hypothesis*}{Hypothesis}

\newtheorem{itheorem}[theorem]{Informal Theorem}

\newtheorem*{claim*}{Claim}

\theoremstyle{definition}
\newtheorem{definition}[theorem]{Definition}

\newtheorem{algorithm}[theorem]{Algorithm}

\newtheorem{assumption}[theorem]{Assumption}

\newtheorem{remark}[theorem]{Remark}
\newtheorem*{remark*}{Remark}

\newtheorem*{observation*}{Observation}

\newcommand{\eat}[1]{}

\renewcommand{\Pr}{\ProbOp}

\renewcommand{\epsilon}{\varepsilon}

\newcommand{\instancespace}{\mathbf{I}}
\newcommand{\solutionspace}{\mathbf{S}}
\newcommand{\md}{\mathbf{d}} 
\newcommand{\alg}{\mathcal{A}}
\newcommand{\radius}{\mathrm{radius}}
\newcommand{\rate}{\mathrm{rate}}
\newcommand{\prev}{\mathrm{prev}}
\newcommand{\mnext}{\mathrm{next}}
\newcommand{\cost}{\mathrm{cost}}

\newcommand{\rc}{\mathrm{rc}}
\newcommand{\hclass}{\mathcal{H}}
\newcommand{\psib}{\Psi_B}

\newcommand{\centers}{\mathbf{C}}
\newcommand{\oracle}{\mathcal{O}}
\newcommand{\distribution}{\mathcal{D}}
\newcommand{\dmax}{\mathbf{d}_{\mathrm{max}}}
\newcommand{\dmin}{\mathbf{d}_{\mathrm{min}}}

\newif\ifnotes\notestrue

\ifnotes
\usepackage{color}
\definecolor{mygrey}{gray}{0.50}
\newcommand{\notename}[2]{{\textcolor{blue}{\footnotesize{\bf (#1:} {#2}{\bf ) }}}}

\else

\newcommand{\notename}[2]{{}}

\fi

\newcommand{\Esymb}{\mathbb{E}}
\newcommand{\Psymb}{\mathbb{P}}

\DeclareMathOperator*{\E}{\Esymb}
 
 \DeclareMathOperator*{\ProbOp}{\Psymb}

 \DeclareMathOperator*{\argmin}{argmin}
 
\usepackage{xspace}
\usepackage{bm}
\usepackage{ifpdf}
\usepackage{shadow,shadethm,color,soul}
\usepackage[compact]{titlesec}
\usepackage{algorithm}
\usepackage[noend]{algorithmic}
\usepackage{subfigure}
\usepackage{verbatim}
\usepackage{paralist}
\usepackage[colorlinks=true,linkcolor=blue,citecolor=blue,allcolors=blue]{hyperref}
\hypersetup{
    colorlinks=true,
    linkcolor=blue,
    filecolor=blue,
    urlcolor=blue,
    allcolors=blue
}
\usepackage[capitalize,nameinlink]{cleveref}

\usepackage{enumerate}
\usepackage{changepage}

\usepackage{tikz}
\usetikzlibrary{arrows}
\usetikzlibrary{positioning}

\allowdisplaybreaks

\definecolor{Darkblue}{rgb}{0,0,0.4}
\definecolor{Brown}{cmyk}{0,0.81,1.,0.60}
\definecolor{Purple}{cmyk}{0.45,0.86,0,0}

\ifpdf
 
\fi


\title{
Competitive strategies to use ``warm start" algorithms with predictions
}
\author{
Vaidehi Srinivas\thanks{Northwestern University: \href{mailto:vaidehi@u.northwestern.edu}{\texttt{vaidehi@u.northwestern.edu}}, supported by the National Science Foundation (NSF) under grants EECS-2216970 and CCF-2154100, and the Northwestern Presidential Fellowship.  Part of this work was done while V.S. was visiting TTIC (Toyota Technological Institute at Chicago), as part of the IDEAL (Institute for Data, Econometrics, Algorithms, and Learning) summer exchange program.}
~and Avrim Blum\thanks{Toyota Technological Institute at Chicago: \href{mailto:avrim@ttic.edu}{\texttt{avrim@ttic.edu}}, supported in part by the National Science Foundation under grants CCF-2212968 and ECCS-2216899.} 
}

\date{}

\begin{document}

\pagenumbering{gobble}

\maketitle
\sloppy

\begin{abstract}

We consider the problem of learning and using predictions for \emph{warm start} algorithms with predictions.  In this setting, an algorithm is given an instance of a problem, and a \emph{prediction} of the solution.  The runtime of the algorithm is bounded by the distance from the predicted solution to the true solution of the instance.   
Previous work has shown that when instances are drawn iid from some distribution, it is possible to learn an approximately optimal fixed prediction \cite{DILMV21}, and in the adversarial online case, it is possible to compete with the best fixed prediction in hindsight \cite{KBTV22}.  

In this work we give competitive guarantees against stronger benchmarks that consider a set of \(k\) predictions \(\mathbf{P}\).  That is, the ``optimal offline cost" to solve an instance with respect to \(\mathbf{P}\) is the distance from the true solution to the \emph{closest} member of \(\mathbf{P}\).  This is analogous to the \emph{\(k\)-medians} objective function.  In the distributional setting, we show a simple strategy that incurs cost that is at most an \(O(k)\) factor worse than the optimal offline cost.  We then show a way to leverage learnable coarse information, in the form of partitions of the instance space into groups of ``similar" instances, that allows us to potentially avoid this \(O(k)\) factor.  

Finally, we consider an online version of the problem, where we compete against offline strategies that are allowed to maintain a \emph{moving} set of \(k\) predictions or \emph{trajectories}, and are charged for how much the predictions move.  We give an algorithm that does at most \(O(k^4 \ln^2 k)\) times as much work as any offline strategy of \(k\) trajectories.  This algorithm is deterministic (robust to an adaptive adversary), and oblivious to the setting of \(k\).  Thus the guarantee holds for all \(k\) simultaneously.  

\end{abstract}

\newpage
\tableofcontents

\newpage
\pagenumbering{arabic}
\section{Introduction}

\emph{Warm start} algorithms are a practically motivated paradigm of algorithms that make use of \emph{predictions} to improve runtime.  It is the most common framework in which \emph{algorithms with predictions} (also \emph{learning-augmented algorithms}, see \cite{BWCA-AlgosWithPredictions}) are studied for static problems.\footnote{By ``static," we mean problems for which we are given all inputs up front, and we wish to minimize runtime. This is as opposed to work in online, streaming, dynamic, and approximation algorithms.}  
A warm start algorithm solves an instance \(I\) from an instance space \(\instancespace\) to find a solution \(S\) in a solution space \(\solutionspace\).  In addition to \(I\), the algorithm is given a prediction \(P \in \solutionspace\), which can be thought of as a ``guess" of the solution.  The runtime of the algorithm depends on how far \(S\) is from \(P\).  This has applications in settings in which one must solve a series of related instances of a problem, and can learn information about ``typical" solutions.  An example is a network routing problem that must be solved daily, for which the user knows that ``today's network traffic is not too different than yesterday's."

In this setting, the solution space \(\solutionspace\) is a metric space with distance \(\md(\cdot, \cdot)\), and the runtime of an algorithm on instance \(I\) with prediction \(P\) can be bounded by \(\md(P, S)\), where \(S\) is the true solution of instance \(I\).  (We wrap factors that depend on the size of the instance and other details into the distance function \(\md\).)  
Examples that fit in this framework include algorithms for bipartite matching \cite{DILMV21, CSVZ22} and network flows \cite{PZ22, DMVW23}. 

Given warm starts as an algorithmic primitive, the question remains how to obtain and use high quality predictions to achieve good performance.  
Previous work on this problem has focused on competing with the best fixed prediction in hindsight.  In the distributional setting, \cite{DILMV21} and \cite{DMVW23} show that it is possible to PAC-learn a good fixed prediction to use with future instances.  In the online setting, \cite{KBTV22} show that it is possible to compete with the best fixed prediction in hindsight, and their work extends to algorithms with predictions in settings beyond warm starts.  

These results demonstrate settings in which warm starts can give us an advantage.  However, to see a large advantage, we must be in a setting where the best fixed prediction leads to much better performance than a generic default prediction.  This means that the solutions of most instances coming from our distribution must be very close to the prediction, and therefore very similar to each other, which is a restrictive assumption.  Our work explores less restrictive structure that we can take advantage of to achieve good performance.  


We observe that warm start algorithms have many properties that we can take advantage of, that make them powerful and flexible.  They are packaged local search routines that allow you to search the solution space by conceptually growing a ball around a prediction, where the radius of a ball corresponds to the runtime of the warm start algorithm. 
It is also possible to run multiple instantiations of the algorithm in parallel, and simultaneously search the space from multiple points.  
Another useful property of warm starts is that every time we solve an instance, we learn the optimal prediction for that instance, i.e.\
the solution.  
These observations motivate our main question:
\begin{quote}
    \center \it Are there settings in which we can use warm-start algorithms to achieve significantly better performance than any single fixed prediction?
\end{quote}

In this work, we leverage these properties of warm starts to compete against a variety of stronger baselines, answering this question in the affirmative.  In particular, we consider competing with a set of \(k\) predictions \(\mathbf{P} \in \solutionspace^k\).  That is, the cost of \(\mathbf{P}\) with respect to an instance \(I\) with solution \(S\) is the distance from \(S\) to the closest prediction: \(\min_{P \in \mathbf{P}} \md(S, P)\).  This is analogous to the \emph{\(k\)-medians} objective.  We also consider competing with the best {\em function} $h$ from some given class, mapping each instance to one of $k$ predictions. 

It is somewhat surprising that it is possible to achieve any non-trivial guarantees against multiple points, especially in the online setting. 
Previous work approaches this problem by reducing to online convex optimization \cite{KBTV22}.  We note that in the standard online convex optimization setting, we cannot expect to compete with a collection of multiple points and achieve vanishing regret.  In particular, even going from 1 trajectory to 2 trajectories makes the offline problem of choosing the best trajectories go from being convex to being non-convex.  
We observe that vanishing regret is actually a stronger property than what we need.  By accruing constant factors in runtime, and taking advantage of structural properties such as running algorithms in parallel, we show that we can compete with far stronger baselines.    

We note that there is a line of work that studies competing against multiple predictions for algorithms with predictions for \emph{online} problems, see \cite{AGKP22, DILMV22, ACEPS23}, which has a different set of challenges and techniques from the warm start setting, which are discussed more in \Cref{subsec:related-work}.  Our work is, to the best of our knowledge, the first to consider this problem for warm starts.  

\paragraph{Competing against multiple points offline.} Our first setting is simple, but we find it beneficial to provide a formal treatment as it motivates our other two settings.  In this setting, we consider instance-solution pairs drawn from a distribution \(\distribution\) over \(\instancespace \times \solutionspace\).  Here, we can use a simple strategy that learns \(k\) fixed predictions for \(\distribution\) from samples.  Then, for a subsequent instance \((I, S) \sim \distribution\), for each prediction we can run an instantiation of our warm start algorithm ``in parallel,"\footnote{By ``in parallel," we refer to interleaving the \(k\) threads of computation, resulting in a sequential algorithm.} and output the solution of the thread that completes first.   

\begin{itheorem}[Competing against \(k\) fixed points offline, Formally \Cref{thm:k-fixed-points}]
    For a distribution \(\distribution\) over \(\instancespace \times \solutionspace\) it is possible to learn a set of \(k\) predictions \(\widehat{\mathbf{P}}\) that is an \(O(1)\) approximation to
    \[\argmin_{\mathbf{P} \in \solutionspace^k} \E_{(I, S) \sim \distribution} \left[ \min_{P \in \mathbf{P}} \md(S, P) \right].\]

    Furthermore, given a set of predictions \(\widehat{\mathbf{P}}\), we can design a procedure such that the time to solve a new instance \(I \in \instancespace\) with (unknown) true solution \(S \in \solutionspace\) is
    \[O(k) \cdot \min_{P \in \widehat{\mathbf{P}}} \md(S, P).\]
\end{itheorem}
We note that a \(\Omega(k)\) approximation factor in the runtime is necessary without assuming more structure in the problem (\Cref{rem:lower-bound}). 

\paragraph{Learning ``coarse information."} This motivates our second setting, in which we avoid this extra \(\Omega(k)\) factor in the runtime.  
Suppose we had access to extra ``coarse information" about the mapping from instances and solutions.  A warm start algorithm, or indeed any algorithm, already encodes the perfect mapping between instances and solutions.  However, this comes at a high cost to uncover.  We model coarse information as a \(k\)-wise partition of the instance space \(h : \instancespace \rightarrow [k]\), where we expect that two instances \(I_1\) and \(I_2\) that map to the same partition have similar solutions.  In addition, we expect \(h(I)\) to be significantly faster to compute than solving the instance \(I\).  Given an \(h\) that satisfies these properties, we can learn a fixed prediction \(P^{(i)}\) for each partition \(i\) of \(h\) with respect to a distribution \(\distribution\) over \(\instancespace \times \solutionspace\).  Then, for a new instance \((I, S) \sim \distribution\), we can use the single prediction \(P^{(h(I))}\).  This allows us to potentially avoid the \(O(k)\) factor in the runtime of the previous theorem.  

Thus the question is whether we can learn a partition \(h\) of \(\instancespace\) that is informative and well-aligned with \(\distribution\).  We show that given a class of efficiently-computable candidate partitions \(\hclass\), it is possible to learn an \(h \in \hclass\) that is approximately best-aligned with the distribution \(\distribution\).


\begin{itheorem}[Competing with a hypothesis class of \(k\)-wise partitions offline, 
Formally \Cref{thm:competing-with-hypothesis-class}]
    For a distribution \(\distribution\) over \(\instancespace \times \solutionspace\), and a learnable hypothesis class \(\hclass\) of \(k\)-wise partitions of \(\instancespace\), given access to an appropriate ERM oracle over \(\hclass\), it is possible to learn an \(h : \instancespace \rightarrow [k]\) and set of predictions \(\mathbf{P} = \left( P^{(1)}, \dots, P^{(k)} \right)\) that are an \(O(1)\) approximation to 
    \[\min_{h \in \hclass} \min_{\mathbf{P} \in \solutionspace^k} \E_{(I, S) \sim \distribution} \left[ \md(S, P^{(h(I))}) \right].\]

    Furthermore, given an \(h : \instancespace \rightarrow [k]\) and \(\mathbf{P} = \left( P^{(1)}, \dots, P^{(k)} \right)\), we can design a procedure such that the expected time to solve a new instance is 
    \[\E_{(I, S) \sim \distribution} \left[ \mathrm{runtime}(h(I)) + \md(S, P^{(h(I))})\right].\]
\end{itheorem}

We achieve this via a two step procedure, that first learns an approximately good set of \(k\) centers for \(\distribution\), then uses an ERM oracle to find an \(h \in \hclass\) that aligns well with this set of centers.  While there may be no reason to expect that the \(k\) centers chosen in the first step are similar to the best centers for the best \(h \in \hclass\), we show that this strategy is still approximately optimal.  

\paragraph{Competitive algorithms to solve sequences of instances.} Finally, we consider an online setting where instances arrive sequentially.  One of the main motivations of warm start algorithms is settings in which instances that arrive close in time are likely to be similar.  One example is the setting explored in experiments by \cite{DMVW23}, in which they use warm starts to speed up solving a series of flow instances that correspond to image segmentation tasks on frames of a video.  

We formulate the \emph{online ball search} problem to model this setting.  In this problem, on each day there is an hidden point \(S \in \solutionspace\) (corresponding to the true solution of that day's instance) that the algorithm must find.  The algorithm is allowed to search the space \(\solutionspace\) by growing balls around points of its choosing in \(\solutionspace\).  The algorithm succeeds when one of the balls first contains \(S\).  Each of these balls corresponds to running an instance of an algorithm \(\alg(I, P)\) using a particular prediction \(P\).  The radius of the ball corresponds to the number of steps for which the algorithm was run.


In this setting, we design competitive algorithms that compete against the best strategy in hindsight.  The strategies that we compete against are collections of \(k\) predictions that are allowed to move over time.  We call these moving predictions \emph{trajectories}.  The cost of a collection of \(k\) trajectories is the total \emph{hit cost} plus the total \emph{movement cost}.  The hit cost on a given day is the distance from that day's solution to the closest of the trajectories on that day.  The total movement cost is the total distance that each of the \(k\) trajectories moves over time.  This is a generalization of competing with \(k\) fixed points, that includes adaptive strategies.  

Note that while the baseline strategies that we compete against are constrained to using \(k\) predictions, and must pay to move the predictions, the online algorithm has no constraint on how many predictions it can use (i.e., points from which it can search), nor does it pay to move predictions.  We note that the restriction on the offline strategies is necessary, because if the offline strategy was allowed to move arbitrarily, it would move perfectly to the solution on each day and achieve cost 0, making it impossible to compete with.

Our objective is to minimize the total runtime to solve a sequence of instances.  The runtime to solve the instances corresponds to the sum of radii that the algorithm searches, which is the number of steps it runs of the subroutine \(\alg\), plus any additional overhead in running the algorithm.

\begin{itheorem}[Competing with trajectories online efficiently, Formally \Cref{thm:online-competitive-with-runtime}]
    We have an algorithm for online ball search that is \(O(k^4 \ln^2 k)\)-competitive with any set of \(k\) trajectories, where the total runtime of the algorithm is bounded by \(O(1)\) times the number of steps it runs of the algorithms-with-predictions subroutine.  That is, this algorithm does at most \(O(k^4 \ln^2 k)\) times as much work as any set of \(k\) trajectories.     

    Furthermore, the algorithm is deterministic and therefore robust to an adaptive adversary, and oblivious to the setting of \(k\), so the guarantee holds for all \(k\) simultaneously.  
\end{itheorem}

This guarantee states that our online algorithm has runtime comparable, within an \(O(k^4 \ln^2 k)\) factor, to any offline strategy that is allowed to maintain \(k\) moving predictions, and pay the cost on each day from that days solution to the \emph{closest} of the predictions.  In some sense, this algorithm does well when the solutions to the instances we see fall into \(k\) (moving) clusters, and our objective is related to the \(k\)-medians objective.  Thus, it is particularly nice that the algorithm does not need to know \(k\), as in practice we often do not know the number of clusters in our data in advance.  

To approach the online problem, it is illustrative to first consider competing with only one trajectory.  In this setting, an algorithm that always uses the previous day's solution as the next day's prediction is 2-competitive (via the triangle inequality).  Note that this simple guarantee already generalizes what was previously known, as it competes not only with the best fixed prediction in hindsight, but also against adaptive strategies that move over time.  For multiple trajectories, this simple strategy is no longer enough, as the previous day's solution could come from some other trajectory, and be arbitrarily far away from today's solution.  

Our algorithm for multiple trajectories addresses this by searching from \emph{all} previous solutions in parallel.  These ``threads" are run at quadratically decaying rates, i.e. the thread of the \(i\)th most recent solution is run at rate \(\frac{1}{i^2 \ln^2 i}\).  This leaves the issue that the previous solution from the same trajectory as today could be arbitrarily far in the past, and be run at an extremely slow rate.  This is addressed by ``pruning" threads that are no longer fruitful.  That is, when the ball searched around a solution \(i\) fully contains the ball searched around a previous solution \(j\), the algorithm can stop running thread \(j\), as that work is redundant with thread \(i\).  Then, we can increase the rates of the threads that are lower priority than \(j\), and maintain that there is at most one thread being run at each rate \(\frac{1}{i^2 \ln^2 i}\) for each integer \(i\).  In the analysis we show that either the algorithm runs long enough for the previous solution from the same trajectory as today to be elevated to rate \(\ge \frac{1}{k^2 \ln^2 k}\) and eventually solve the instance, or the algorithm terminates more quickly than that, which is even better.  This allows us to bound the competitive ratio.  

To bound the runtime, we note that the amount of work done by each thread is actually dominated by the work done to check if it was been pruned yet.  This blows up the work done by the \(i\)th fastest thread by a factor \(O(i)\).  Since the rate at which the \(i\)th thread takes steps of \(\alg\) is \(\frac{1}{i^2 \ln^2 i}\), the rate at which it does work is \(O(i) \cdot \frac{1}{i^2 \ln^2 i} = O(\frac{1}{i \ln^2 i})\).  Thus, summing the work done over all of the threads results in a convergent series, and we are able to bound the total work done by the algorithm by the work done by the fastest thread.  

We remark that if we relax our objective to only be competitive in sum of radii searched, and not worry about the runtime of the algorithm, the problem is still interesting.  This could correspond to settings where steps of the subroutine are much more costly than steps of the online scheduling algorithm.  In this setting, we can design an algorithm with an improved competitive ratio of \(O(k^2)\) based on a reduction to the \(k\)-server problem (\Cref{cor:online-ball-search-via-k-server}).    

In this work, we demonstrate that warm starts are a powerful algorithmic primitive that can, in many settings, be used in ways that are much stronger than competing with one fixed point in hindsight.  This opens new research directions in more effective ways to learn and use warm starts.  We also observe that the properties of warm starts that we use may hold in many settings that use local-search type algorithms, and we hope that our techniques can be extended to broad-ranging applications.





\subsection{Related Work}
\label{subsec:related-work}

\paragraph{Algorithms with predictions.} \emph{Algorithms with predictions} (also \emph{learning-augmented algorithms}) is a \emph{beyond worst-case} paradigm of algorithm design that has become well-studied in recent years.  In this paradigm, an algorithm solicits an untrusted \emph{prediction} to help solve a worst-case instance.  Generally speaking, the goals are to provide (i) \emph{consistency}: better performance than a worst-case algorithm when the prediction is of high quality, (ii) \emph{robustness}: performance no worse than a worst-case algorithm when the prediction is of low quality, and (iii) \emph{graceful degradation}: performance degrades smoothly as a function of the quality of the predictions.  The type of prediction and how we measure prediction error can vary vastly among different problem settings, as well as the type of performance that we are trying to optimize (e.g.\ runtime for warm starts, competitive ratio for online algorithms, memory usage for streaming algorithms).  Thus the algorithms and techniques that have been developed for algorithms with predictions are also diverse.  For an overview of algorithms with predictions, the reader is referred to the book chapter of Mitzenmacher and Vassilvitskii \cite{BWCA-AlgosWithPredictions}, with the note that the body of work in this area has grown significantly even in the few years since this was published.

\paragraph{Data-driven algorithms.} \emph{Data-driven algorithm design} is a beyond worst-case paradigm of algorithm design that is very related to algorithms with predictions.  
In this setting, we consider a \emph{parametrized family} of algorithms for a certain problem, and we wish to learn the best setting of parameters for instances drawn from a distribution \(\distribution\) (or to achieve low regret in the online case).  Work in this area typically focuses on proving learning guarantees, e.g.,\ how many samples needed to learn an approximately-optimal setting of parameters.  For an overview of data-driven algorithms, the reader is referred to the book chapter of Balcan \cite{BWCA-DataDrivenAlgorithms}.

We can think of this as a different viewpoint for algorithms with predictions.  For an algorithm with predictions \(\alg(I, P)\) where we can think of predictions as parameters, and each possible prediction \(P\) defines an algorithm in the family, i.e.\ the family of algorithms \(\{\alg_P(I) = \alg(I, P) \mid P \in \solutionspace\}\).  Work in algorithms with predictions typically focuses on showing that the performance of the algorithms in the family can be significantly better than the performance of a worst-case algorithm.  Ideally, we would like to have both types of guarantees: (i) that we can learn good predictions, and (ii) that good predictions enable us to achieve far better performance.  Our work attempts to bridge this gap by providing guarantee (i) to problems for which we have guarantee (ii).

Data-driven algorithms are often studied in online settings where one must learn a good parameter setting over time.  This is usually studied in the framework of regret, where an algorithm must compete with the best fixed parameter setting in hindsight.  There is a line of work that considers ``mixture" settings, in which the performance of the algorithm competes against multiple points in hindsight.  
\cite{SBD20} gives guarantees for  \emph{shifting regret}, where the algorithm must compete against offline strategies that can switch parameter settings some fixed number of times.    \cite{BKST21} studies \emph{meta-learning}, in which an algorithm sees samples from multiple distributions one at a time, and uses information from previous distributions to learn good initializations for future distributions.  These settings are different from ours, as the set of points that they compete with are considered sequentially, whereas we compete with a \(k\)-medians style objective.  \cite{KCBT24} uses a \emph{contextual bandit} framework to choose different parameter settings for different types of instances.  This is related to the setting that we study in \Cref{sec:k-wise-partitions}, in that we are learning information about the instance space, to provide parameter settings that are more fit to each instance.

\paragraph{Warm starts.}  \emph{Warm starts} are a popular heuristic that are used in practice to speed up the computation of various problems, e.g.\ for linear and mixed integer programs \cite{gurobi}.  Due to their success in practice, a line of work has sought to provide rigorous theoretical guarantees for problems such as bipartite matching \cite{DILMV21, CSVZ22} and max flow \cite{PZ22, DMVW23}.
Previous work has studied the problem of solving a sequence of instances using a learning-augmented algorithm to compete with the performance of the best fixed prediction in hindsight \cite{KBTV22}.  Our work is, to the best of our knowledge, the first to consider the problem of competing with multiple predictions in the warm start setting.

\paragraph{Dynamic algorithms.}  We note that the online setting of our problem, in which we solve a sequence of related instances, is related to the setting of \emph{dynamic algorithms}.  In dynamic algorithms, the input can change in small, structured ways from one day to the next (e.g.\ for a graph, one edge is inserted or deleted per day).  A dynamic algorithm must update its solution based on the change in the input, and seeks to minimize its \emph{update time}, or runtime to perform this update on a given day.  In dynamic algorithms, we typically hope to achieve update time that is sublinear, or even much smaller, in the size of the input.  This justifies that a dynamic algorithm is much more efficient than solving each day's instance from scratch.  

In the warm start setting on the other hand, the input instance can change arbitrarily from day to day.  Thus, we cannot hope for sublinear update times, as the algorithm must at least read the input and verify the predicted solution on each day.  However, as we show in this work, in this less structured setting, it is possible to compete against stronger baselines.  For example, consider a setting where on each day we receive an input corresponding to one of several slowly changing graphs.  This falls outside the standard dynamic model, but can be addressed by our guarantee that competes with multiple trajectories.  
There is also a line of work that considers \emph{dynamic algorithms with predictions} \cite{vdBFNP24, HSSY24, LS23, AB24}, but this setting is quite different than the one we consider in this work.

\paragraph{Multiple predictions for online algorithms.} There is a line of work that studies competing with multiple predictions for \emph{online algorithms with predictions}.  Online algorithms usually have some notion of commitment, where on each day the algorithm must make an irrevocable decision that affects future performance.  Typically, algorithms with predictions in this setting solicit information about future events.  Given multiple such predictions, the challenge is combining them into \emph{one} decision that can be made on a given time step.  Because the decisions that the algorithm makes vary from problem to problem, techniques to combine predictions often have to be problem specific.  Work in this area includes strategies for scheduling problems \cite{DILMV22}, set cover and facility location \cite{AGKP22}, and metric algorithms \cite{ACEPS23}.  This setting is significantly different from ours.  In our setting, we can asymptotically compete with multiple predictions by running multiple instantiations of the learning-augmented algorithm in parallel.  This leads to a different set of strategies and techniques than in the online setting.

\paragraph{Online search and server problems.}  The \emph{online ball search} problem that we formulate in this work has connections to other well-studied online problems.  
Previous work has approached this problem through the lens of \emph{online convex optimization} \cite{KBTV22}, to provide competitive guarantees against a single fixed point in hindsight.  
We note that in going from competing against one trajectory to competing against multiple trajectories, our problem becomes non-convex, and we require a different set of tools to approach it.  

For competing against \(k\)-trajectories, the closest connection is to the \(k\)-server problem \cite{MMS88, MMD90}, in which we must service a sequence of \emph{requests} in a metric space, using \(k\) \emph{servers}, while moving the servers as little as possible over the course of the algorithm.  In our setting, we can think predictions as servers, and the solutions of arriving instances as being requests, a connection that we explore in \Cref{subsec:reduction-to-k-server}.  While we can solve online ball search with a reduction to the \(k\)-server problem, it is not clear if they are equivalent problems, and another approach we propose for online ball search (\Cref{subsec:improved-runtime}) does not go through this connection.  

A related online search problem is the \emph{oil searching problem} \cite{MOP09}, in which an algorithm can put in work to search \(n\) locations to certain depths to try to find a resource.  This is related to our problem, in which the algorithm can put in work at various predictions to try to find a solution near that prediction.  However, our problem allows the algorithm more freedom in choosing the staring location, and has a different objective function.  Other related search problems are the \emph{cow path problem} \cite{KRT93}, and the problem of \emph{searching in the plane} \cite{BCR93}.  In these problems the algorithm is allowed to traverse the space in different ways than in online ball search.  

The objective of the online ball search problem is analogous to the \(k\)-medians objective function in clustering, in some sense.  Thus, the online problem also has connections to online and dynamic \emph{\(k\)-medians clustering} \cite{BCLP23}, and \emph{online steiner tree} \cite{IW91, AA92, GGK16} though the objectives of these problems are somewhat different from ours.  We note that the online steiner tree problem has also been studied in the learning-augmented setting \cite{XM21}, a setting which is quite different from this work.

\section{Preliminaries}
\label{sec:preliminaries}

We model an algorithm as solving instances from an instance space \(\instancespace\) to produce solutions from a solution space \(\solutionspace\).  That is, a (standard) algorithm \(\alg\) takes \(I \in \instancespace\) as input to produce \(\alg(I) = S \in \solutionspace\).  
To model a warm start algorithm, we assume a metric distance \(\md\) on the solution space \(\solutionspace\). 

\begin{definition}[Warm start algorithm]
    A \emph{warm start} algorithm \(\alg\) is one that takes a problem instance \(I \in \instancespace\) and a \emph{predicted solution} \(P \in \solutionspace\) as input.  \(\alg(I, P)\) outputs the true solution \(S \in \solutionspace\) of instance \(I\) in time \(\le \md(P, S)\).
\end{definition}

To justify modelling the runtime of an algorithm with predictions as a metric distance over the solution space \(\solutionspace\), consider the following examples.  For bipartite matching, \cite{DILMV21} give an algorithm with runtime \(\widetilde{O} \left(m \sqrt{n} \cdot (1 + ||y^* - \widehat{y}||_1) \right)\), where \(\widehat{y}\) is the predicted (dual) solution, and \(y^*\) is the optimal (dual) solution.  \cite{CSVZ22} improve the guarantee to \(O \left(m \sqrt{n} + (m + n \log n)||y^* - \widehat{y}||_0 \right)\).  Both of these runtimes can be interpreted as a metric distance that is essentially the distance between the true solution and the predicted solution in the relevant norm, scaled by a factor that depends on the input size.  Another example is the algorithm of \cite{DMVW23} which solves instances of max-flow in time \(O (|E| \cdot (1 + ||\widehat{f} - f^*||_1) )\), where \(|E|\) is the number of edges in the flow network, \(\widehat{f}\) is the predicted flow (as a vector), and \(f^*\) is the optimal flow closest to \(\widehat{f}\).

Note that we consider copies of the same solution to be distinct entities at some fixed distance from each other that depends on the input size.  This is because even when the predicted solution is exactly the true solution, a warm-start algorithm will still take some amount of time to verify the solution.  This is reflected in the time-bounds that are given above, which are lower bounded by a constant that depends on the input size. This is consistent with our characterization of the runtime as a metric on the solution space.  For example, in many settings, we can take the distance between a prediction \(P\) and a solution \(S\) to be 
\[\md(P, S) = |I| \cdot (1 + ||P - S||)\] 
where \(|I|\) is the size of the relevant instances, and \(||P - S||\) is the distance between \(P\) and \(S\) as vectors taken in the relevant norm.

In the remainder of the paper, we will assume that the runtime of a warm start algorithm corresponds to some metric distance \(\md\) that is known and easily computable.  In particular, we make the following assumption. 

\begin{assumption}[Metric distance is easily computable]\label{ass:metric-distance-computable}
    It is possible to compute the distance, \(\md(S_1, S_2)\), between any two points, \(S_1, S_2 \in \solutionspace\), in time at most \(\dmin\), where \(\dmin\) is the minimum distance in \(\solutionspace\).  Recall that the distance \(\solutionspace\) is lower bounded, as we consider two copies of the same point to be distinct entities at some fixed distance away from each other.  

    For ease of notation, we will consider the distance \(\md\) to be scaled by a factor such that \(\dmin = 1\) and the time to compute the distance between points is \(O(1)\).  This is without loss of generality, as we give multiplicative runtime guarantees in terms of \(\md\). 
\end{assumption}

This assumption is reasonable, because in our applications of interest, the points in \(\solutionspace\) are represented by vectors of dimension that scales with the input size, and the distance function is distance taken in an appropriate norm.  The runtime of the warm-start algorithms also has a multiplicative factor that scales with the input size, so the time to calculate the distance is on the same order of magnitude as the time to perform ``one iteration" of the warm-start algorithm.  

\section{Competing against \(k\) fixed points offline}
\label{sec:k-fixed-points}

As a warm up and motivating example, we consider a distribution \(\distribution\) over instance-solution pairs \(\instancespace \times \solutionspace\). We assume that there is a underlying ground-truth mapping from instances \(I \in \instancespace\) to solutions \(S \in \solutionspace\) that is uncovered by our algorithm-with-predictions \(\alg\).  Thus it is sufficient to assume a distribution only over \(\instancespace\), and we include the solution in the distribution for ease of notation.  In many places it will be useful to consider the marginal distribution of solutions. 

We show that, with access to i.i.d.\ samples from \(\distribution\), a simple strategy of running the predictions in parallel can compete with the best \(k\) fixed predictions for the distribution, with an \(O(k)\) approximation factor.  This is a generalization of the results in \cite{DILMV21} and \cite{DMVW23}, which show that it is possible to PAC learn the single best fixed prediction with respect to \(\distribution\).

Our main observation is that a warm start algorithm-with-predictions allows us to run multiple threads in parallel, allowing us to compete with the performance of the best thread.


\begin{lemma}[Using \(k\) predictions]
Given an algorithm-with-predictions \(\alg\) for instances in \(\instancespace\), and a set of \(k\) predictions \(\mathbf{P} = (P_1, \dots, P_k), P_i \in \mathbf{S}\), for an instance \(I\) with (unknown) true solution \(S\), we can solve \(I\) in time 
\[O(k \cdot \md(S, \mathbf{P}(S))),\]
where \(\mathbf{P}(S) = \argmin_{P_j \in \mathbf{P}} \md(S, P_j) \).
\label{lem:using-k-predictions-in-parallel}
\end{lemma}

\begin{proof}
We run the algorithm-with-predictions in parallel with each of the \(k\) predictions, and output the solution of the thread that completes first (\Cref{alg:k-predictions-parallel}).  Let \(\mathbf{P}(S) = P_{j^*}\) be the prediction that minimizes \(\md(S, P_{j^*})\).  We know that the time that \Cref{alg:k-predictions-parallel} spends on thread \(j^*\), i.e. the runtime of \(\mathcal{A}(I, P_{j^*})\), is \(O(\md(S, P_{j^*}))\).  Since the threads are run in parallel at the same rate, this means that the time that \Cref{alg:k-predictions-parallel} spends on any thread \(j\) is \(O(\md(S, P_{j^*}))\), and therefore the total runtime of the algorithm is bounded by \(O(k \cdot \md(S, P_{j^*}))\).  
\end{proof}

\begin{algorithm}
\caption{Using \(k\) predictions in parallel}
\label{alg:k-predictions-parallel}
\begin{algorithmic}[1]
    \STATE Run \(\mathcal{A}(I, P_j)\) for all \(j \in [k]\) ``in parallel"\footnotemark \ until one of them completes
    \STATE Output the solution of the \(\mathcal{A}(I, P_j)\) that completed
\end{algorithmic}
\end{algorithm}
\footnotetext{Here, we use ``in parallel" to mean alternately running a constant number of steps of each of the \(\mathcal{A}(I, P_j)\), resulting in a sequential algorithm.}

We observe that it is possible to learn an approximately good set of predictions with respect to a distribution \(\distribution\) over \(\solutionspace\) from samples, by noting that this is the \(k\)-medians clustering problem for the distribution \(\distribution\).
We provide a standard sample compression argument that assumes little structure on \(\solutionspace\) for completeness, and note that better bounds are known for specific settings of interest. 

\begin{definition}[Clustering cost of a set]
    Let \(\mathbf{X} = \{X_1, \dots, X_m\}\) be a set of \(m\) points from \(\solutionspace\), and \(\centers \in \solutionspace^k\) be an arbitrary set of \(k\) centers.  We define the \emph{cost} of \(\centers\) over \(\mathbf{X}\) by
    \[\cost(\centers, \mathbf{X}) = \frac{1}{m} \sum_{i = 1}^m \left[ \md(\centers(X_i), X_i) \right],\]
    where \(\centers(X_i)\) is the closest center in \(\centers\) to \(X_i\).
\end{definition}

\begin{definition}[Clustering cost of a distribution]
    Let \(\distribution\) be a distribution over \(\instancespace \times \solutionspace\), and \(\centers \in \solutionspace^k\) be an arbitrary set of \(k\) centers.  We define the \emph{cost} of \(\centers\) over \(\distribution\) by
    \[\cost(\centers, \distribution) = \E_{(I, S) \sim \distribution} \left[ \md(\centers(S), S) \right],\]
    where \(\centers(S)\) is the closest center in \(\centers\) to \(S\).
    \label{def:clustering-opt}
\end{definition}


\begin{restatable}[Learning \(k\) fixed points]{lemma}{learningkfixedpoints}
    For a distribution \(\distribution\) over instance-solution pairs, it is possible to learn an \(O(1)\)-approximate \(k\)-medians clustering of \(\distribution\) from \(m \ge \frac{12 \cdot \dmax \cdot k \cdot \log(1/\delta)}{\cost(\centers^*, \distribution)} \) samples drawn i.i.d. from \(\distribution\), with probability \(\ge 1 - 2\delta\), where \(\dmax\) is the width of \(\solutionspace\).
    \label{lem:learning-k-fixed-points}
\end{restatable}

\begin{proof}
    Let \(\centers^* = (\centers^*_{(1)}, \dots, \centers^*_{(k)})\) be a best set of \(k\) centers for distribution \(\distribution\).  Let \(\dmax\) be the width (largest distance) of the metric space \(\solutionspace\).  Consider \(\mathbf{X} = \{X_1, \dots, X_m\}\), a set of \(m\) samples drawn i.i.d.\ from \(\distribution\).  First, we bound the probability that the average loss of \(\centers^*\) over \(\mathbf{X}\) is far from the average loss of \(\centers^*\) over \(\distribution\).  Let 
    \[\cost(\centers^*, \distribution) = \E_{X \sim \distribution} \left[ \md(\centers^*(X), X) \right],\]
    where \(\centers^*(X)\) is the closest \(C \in \centers^*\) to \(X\).
    Similarly, define the empirical cost of \(\centers^*\) over \(\mathbf{X}\) as  
    \[\cost(\centers^*, \mathbf{X}) = \frac{1}{m} \sum_{i = 1}^m \md(\centers^*(X_i), X_i).\]
    Now we can use a Chernoff-Hoeffding bound to see that
    \begin{align*}
        \ProbOp \left[ \cost(\centers^*, X) \ge 2 \cdot \cost(\centers^*, \distribution) \right] 
        &= e^{-\frac{m \cdot \cost(\centers^*, \distribution)}{3 \dmax}} .
    \end{align*}
    Thus, by setting \(m \ge \log(1/\delta) \cdot \frac{3 \dmax}{\cost(\centers^*, \distribution)}\), for some \(0 < \delta < 1\) we can achieve that the empirical loss of \(\centers^*\) over the samples \(X\) is at most twice the true loss over the distribution, with probability \(\ge 1 - \delta\).

    Now, we show that it is possible to learn an approximately good clustering for \(\distribution\) from \(\mathbf{X}\).  Our learning algorithm proceeds as follows.  It considers all \(\binom{m}{k}\) subsets of \(\mathbf{X}\) as possible centers.  Of the possible centers, it chooses the \(\widehat{\centers}\) with the lowest empirical loss.  

    By a sample compression argument, see e.g.\ Theorem 30.2 in \cite{SB14}, we have that with probability at least \(1 - \delta\)
    \[\cost(\widehat{\centers}, \distribution) \le \cost(\widehat{\centers}, \mathbf{X} \setminus \widehat{\centers}) + \sqrt{\cost(\widehat{\centers}, \mathbf{X} \setminus \widehat{\centers}) \frac{4k \log(m/\delta)}{m}} + \frac{8k\log(m/\delta)}{m}.\]
    Setting \(m\) to be sufficiently larger than \(4k \log(1 /\delta) \dmax\) allows us to conclude that 
    \[\cost(\widehat{\centers}, \distribution) \le 2 \cdot \cost(\widehat{\centers}, \mathbf{X} \setminus \widehat{\centers}).\]

    Now, we note that there is a clustering of \(\mathbf{X}\) that uses points from \(\mathbf{X}\) as centers, that has cost at most 2 times that of \(\centers^*\) on \(\mathbf{X}\).  To see this, consider mapping each center in \(\centers^*\) to its nearest point in \(\mathbf{X}\), to get \(\centers'\), 
    \[\centers'_{(i)} = \argmin_{X \in \mathbf{X}} \md(\centers^*_{(i)}, X).\]
    We have that for each \(X \in \mathbf{X}\)
    \begin{align*}
        \md(\centers'(X), X) &\le \md(\centers^*(X), X) + \min_{X' \in \mathbf{X}} \md(\centers^*(X), X') \\
        &\le 2 \cdot \md(\centers^*(X), X).
    \end{align*}
    Since \(\widehat{\centers}\) is the cost minimizer over all sets of centers that are subsets of \(\mathbf{X}\), this means that \(\cost(\widehat{\centers}, \mathbf{X}) \le 2 \cdot \cost(\centers^*, \mathbf{X}))\). Furthermore, since \(2k \le m\), we have that \(\cost(\widehat{\centers}, \mathbf{X} \setminus \widehat{\centers}) \le 2 \cdot \cost(\widehat{\centers}, \mathbf{X})\).
    Finally, we get that if \(m \ge \frac{12 \cdot \dmax \cdot k \cdot \log(1/\delta)}{\cost(\centers^*, \distribution)} \), then with probability \(\ge 1 - 2\delta\)
    \begin{align*}
        \cost(\widehat{\centers}, \distribution) &\le 2 \cdot \cost(\widehat{\centers}, \mathbf{X} \setminus \widehat{\centers}) \\
        &\le 4 \cdot \cost (\widehat{\centers}, \mathbf{X} ) \\
        &\le 8 \cdot \cost (\centers^*, \mathbf{X}) \\
        &\le 16 \cdot \cost(\centers^*, \distribution),
    \end{align*}
    so \(\widehat{\centers}\) is an \(O(1)\)-approximation to the minimum cost \(k\)-clustering of \(\distribution\).
\end{proof}

Using the above lemmas, we can conclude the following theorem.

\begin{theorem}[Competing against \(k\) fixed points offline]
    With access to \(m \ge \frac{12 \cdot \dmax \cdot k \cdot \log(1/\delta)}{\cost(\centers^*, \distribution)} \) i.i.d. samples from \(\distribution\), it is possible to design an algorithm that has expected runtime 
    \[O(k) \cdot \cost(C^*, \distribution)\]
    on future instances drawn from \(\distribution\), where \(\centers^*\) is the set of \(k\) centers in \(\solutionspace\) that minimizes the cost with respect to \(\distribution\).
    \label{thm:k-fixed-points}
\end{theorem}

\begin{proof}
    \Cref{lem:learning-k-fixed-points} tells us that from \(m \ge \frac{12 \cdot \dmax \cdot k \cdot \log(1/\delta)}{\cost(\centers^*, \distribution)} \) samples, with probability \(\ge 1 - 2 \delta\) we can learn a set of centers \(\widehat{\centers}\) such that 
    \[\cost(\widehat{\centers}, \distribution) = O(1) \cdot \cost(\centers^*, \distribution).\]
    Once we have \(\widehat{\centers}\), on a subsequent instance-solution pair \((I, S) \sim \distribution\), we can use the algorithm from \Cref{lem:using-k-predictions-in-parallel} to solve \(I\) in expected time 
    \begin{align*}
        \E_{(I, S) \sim \distribution} \left[ O(k \cdot \md(S, \widehat{\centers}(S))\right] &= O(k) \cdot \cost(\widehat{\centers}, \distribution) \\
        &= O(k) \cdot \cost(\centers^*, \distribution).
    \end{align*}
\end{proof}

\begin{remark}
    In the generality that we have modeled the problem in this section, we cannot hope to achieve an approximation factor that is \(o(k)\).  Consider the following bad example.  There are \(k\) planted solutions that are arbitrarily far apart.  Our distribution serves instances such that there is a uniform probability that any of these planted solutions is the true solution.  Thus, an algorithm must either explore a constant fraction of the \(k\) planted solutions in expectation, incurring a \(O(k)\)-approximation, or pay the large distance between the solutions, resulting in an unbounded approximation.  

    It is potentially possible to get around this lower bound by taking advantage of additional structure in the problem.  We explore one such approach in \Cref{sec:k-wise-partitions}.
    \label{rem:lower-bound}
\end{remark}

\section{Competing with a hypothesis class of \(k\)-wise partitions offline}
\label{sec:k-wise-partitions}

In the previous section, we showed that in the offline setting, it is possible to construct an algorithm that competes with the best \(k\)-wise clustering cost for the distribution, with a factor \(O(k)\) blow-up.  We also show that, without additional assumptions, this is the best approximation factor achievable for this model.  
In this section, we investigate ways to leverage extra information to remove the multiplicative \(O(k)\) factor.  In particular, to achieve better performance than the lower bound, we must have some additional way to extract information about the location of \(S\) from \(I\).  

A learning-augmented algorithm, or indeed any algorithm to solve instances \(I\), already encodes the mapping from instances to solutions.  However, it encodes the exact mapping, and the cost of uncovering the mapping is high.  It is reasonable that for many distributions of interest over \(\instancespace \times \solutionspace\), we can learn a \emph{coarse} mapping between instances and solutions, that can help us search for the solution faster.  

A natural way to model a coarse mapping between instances and solutions is as a \(k\)-wise partition over the instance space \(\instancespace\), i.e.\ \(h : \instancespace \rightarrow [k]\), where we expect that instances in a particular partition have solutions that are similar.
With this additional assumption of access to a partition \(h : \instancespace \rightarrow [k]\), we can learn the best fixed prediction for each partition of \(h\).  Once we have the \(k\) predictions fixed, for a new instance \(I\), we can evaluate \(h(I)\) and use the corresponding prediction to solve \(I\).  This avoids the factor \(k\) blow up in runtime that we had to incur in the previous section.  

We note that this is somewhat different from the approach in the previous section (\Cref{sec:k-fixed-points}).  In the previous section, we can think of the \(k\) predictions \(\centers\) that we choose as defining an implicit partition of the solution space, where each partition corresponds to a section in the Voronoi diagram of \(\solutionspace\) defined by \(\centers\).  Because the partition is over the solution space, given a fresh instance \(I\), it is not easy to see which partition \(I\)'s solution belongs to.   Thus we must run all of the predictions, and accrue an \(\Omega(k)\) approximation factor.  In this section, we are considering partitions \(h\) over \emph{instance} space, so that given a new instance \(I\), we can easily compute the partition that \(I\) belongs to.  This also implicitly defines subsets of the solution space \(\solutionspace\), where the \(i\)th (potentially overlapping) subset consists of solutions corresponding to instances in the \(i\)th partition of \(h\).  We note that these subsets may not have any particular structure.  

The question remains whether it is possible to learn a good partition \(h : \instancespace \rightarrow [k]\) for a distribution \(\distribution\).  We consider the setting where we must select a partition from a hypothesis class \(\hclass\) of potential efficiently-computable partitions of \(\instancespace\).  We aim to minimize the expected runtime on future instance-solution pairs from \(\distribution\).  We show that when the hypothesis class \(\hclass\) of \(k\)-wise partitions of \(\instancespace\) is learnable, 
it is indeed possible to learn an approximately optimal \(h \in \hclass\) and set of predictions, or ``centers," \(\centers \in \solutionspace^k\) for the partitions of \(h\).  

We do this by constructing a loss function, \(\centers\)-loss that can be used with an ERM oracle.  In \Cref{subsec:partitions-approximation}, we give the main result of this section, which is showing that the \(\centers\)-loss of a partition \(h : \instancespace \rightarrow [k]\) approximates the cost of \(h\) over a distribution \(\distribution\).  In \Cref{subsec:partitions-learning} we show that if the hypothesis class \(\hclass\) of partitions is learnable, then we can learn an approximately optimal hypothesis \(h \in \hclass\).  In \Cref{subsec:partitions-guarantee}, we conclude an algorithmic guarantee.  

\begin{definition}[Clustering cost of \(k\)-wise partition]
    Let \(\distribution\) be a distribution over \(\instancespace \times \solutionspace\), and \(h\) be a \(k\)-wise partition of \(\solutionspace\).  We define the \emph{cost} of \(h\) over \(\distribution\) by
    \[\cost(h, \distribution) = \E_{(I, S) \sim \distribution} \left[ \md(\centers_h^{(h(I))}, S) \right],\]
    where \(\centers_h\) is the collection of the best centers for each partition of \(h\), with respect to \(\distribution\).  That is, 
    \[\centers_h^{(i)} = \argmin_{C \in \solutionspace} \E_{(I, S) \sim \distribution} \left[ \md(C, S) \mid h(I) = i \right].\]
    \label{def:clustering-cost-of-partition}
\end{definition}

\begin{definition}[\(\centers\)-loss]
    Let \(\distribution\) be a distribution over \(\instancespace \times \solutionspace\), and \(h\) be a \(k\)-wise partition of \(\solutionspace\), and \(\centers \in \solutionspace^k\) be an arbitrary set of \(k\) centers.  We define the \emph{\(\centers\)-loss} of \(h\) on a point \((I, S) \in \instancespace \times \solutionspace\) as 
    \[\ell_\centers (h, (I, S)) = \md(S, \centers^{(h(I))}).\]
    We denote the expected \(\centers\)-loss of \(h\) over \(\distribution\) by
    \[\ell_{\centers}(h, \distribution) = \E_{(I, S) \sim \distribution} \left[ \md(S, \centers^{(h(I))}) \right].\]
    \label{def:c-loss}
\end{definition}

\begin{definition}[Rotation]
    For \(h \in \hclass\), where \(\hclass\) is a hypothesis class of \(k\)-wise partitions, we say that the \emph{rotation} of \(h\) by some \(\varphi : [k] \rightarrow [k]\) is \(\varphi \circ h\).  (Note that \(\varphi\) does not have to be bijective.)
\end{definition}

\begin{definition}[Rotational completion]
    We say that a hypothesis class \(\hclass\) of \(k\)-wise partitions is \emph{rotationally complete} if for every \(h \in \hclass\) and \(\varphi : [k] \rightarrow [k]\), 
    \[ \varphi \circ h \in \hclass.\]
    
    For any hypothesis class \(\hclass\) of \(k\)-wise partitions, we denote the \emph{rotational completion} of \(\hclass\) by \(\rc(\hclass)\), where we define
    \[\rc(\hclass) = \{\varphi \circ h \mid h \in \hclass, \varphi : [k] \rightarrow [k]\} .\]
\end{definition}

\subsection{Approximation}
\label{subsec:partitions-approximation}

The main technical component of this section is designing the \(\centers\)-loss (\Cref{def:c-loss}), that approximates the clustering cost of a \(k\)-wise partition, while also being decomposable and therefore usable with an ERM oracle.  

\begin{lemma}[\(\ell_{\centers}(h, \distribution)\) approximates clustering cost]
    Consider a distribution \(\distribution\) over pairs \((I, S) \in (\instancespace, \solutionspace)\), an arbitrary set of \(k\) centers \(\centers = (\centers^{(1)}, \dots, \centers^{(k)}) \in \solutionspace^k\), and a hypothesis class \(\hclass\) of \(k\)-wise partitions of \(\solutionspace\).  For every \(h \in \hclass\), there exists a rotation \(\varphi : [k] \rightarrow [k]\), such that 
    \[\ell_{\centers}(\varphi \circ h, \distribution) \le O(1) \cdot \left(\cost(h, \distribution) + \cost(\centers, \distribution)\right).\]
    \label{lem:ERM-approximation}
\end{lemma}

\begin{proof}
    First, we consider the cost associated with a particular partition \(i\) of \(h\).  We show how to choose a center \(j\) from \(\centers\) to assign this partition.  Let \(\centers^{(i)}_h\) be the best possible center in \(\solutionspace\) for partition \(i\) of \(h\) for \(\distribution\).  That is, 
    \[\centers^{(i)}_h = \argmin_{C \in \solutionspace} \E_{(I, S) \sim \distribution} \left[ \md(S, C) \mid h(I) = i \right].\]
    Let \(\centers(S)\) be the best center from \(\centers\) for a solution \(S \in \solutionspace\).  That is, 
    \[\centers(S) = \argmin_{\centers^{(j)}} \md(S, \centers^{(j)}).\]
    Now, we can associate \(\centers^{(i)}_h\) with its closest center in \(\centers\).  Formally, let 
    \[\varphi(i) = \argmin_{j \in [k]} \md(\centers^{(j)}, \centers^{(i)}_h) .\]  
    Then we have
    \begin{align*}
        &\E_{(I, S) \sim \distribution} [ \md(S, \centers^{(\varphi(i))}) \mid h(I) = i ] \\
        \le& \E_{(I, S) \sim \distribution} [ \md(S, \centers^{(i)}_h) \mid h(I) = i ] + \md(\centers^{(\varphi(i))}, \centers^{(i)}_h) &\text{triangle inequality} \\ 
        \le& \E_{(I, S) \sim \distribution} [ \md(S, \centers^{(i)}_h) \mid h(I) = i ] + \E_{(I, S) \sim \distribution} [ \md(\centers(S), \centers^{(i)}_h) \mid h(I) = i ] \\
        \le& \E_{(I, S) \sim \distribution} [ \md(S, \centers^{(i)}_h) \mid h(I) = i ] + \E_{(I, S) \sim \distribution} [ \md(S, \centers^{(i)}_h) + \md(S, \centers(S)) \mid h(I) = i ] &\text{triangle inequality}\\
        =& 2 \cdot \E_{(I, S) \sim \distribution} [ \md(S, \centers^{(i)}_h) \mid h(I) = i ] + \E_{(I, S) \sim \distribution} [ \md(S, \centers(S)) \mid h(I) = i ]
    \end{align*}
    Taking the appropriate combination over the partitions \(i\), we get 
    \begin{align*}
        & \ell_{\centers}(\varphi \circ h, \distribution) \\
        =&\E_{(I, S) \sim \distribution} [\md(S, \centers^{(\varphi(h(I)))})] \\
        =& \sum_{i \in [k]} \Pr_{(I, S) \sim \distribution}[h(I) = i] \cdot \E_{(I, S) \sim \distribution} [ \md(S, \centers^{(\varphi(i))})) \mid h(I) = i ]\\
        \le& \sum_{i \in [k]} \Pr_{(I, S) \sim \distribution}[h(I) = i] \cdot \left( 2 \cdot \E_{(I, S) \sim \distribution} [ \md(S, \centers^{(i)}_h) \mid h(I) = i ] + \E_{(I, S) \sim \distribution} [ \md(S, \centers(S)) \mid h(I) = i ] \right) \\
        =& 2 \cdot \E_{(I, S) \sim \distribution} [ \md(S, \centers^{(h(I))}_{h})] +  \E_{(I, S) \sim \distribution} [ \md(S, \centers(S)) ] \\
        =& 2 \cdot \cost(h, \distribution) + \cost(\centers, \distribution).
    \end{align*}
\end{proof}

\subsection{Learning the hypothesis class}
\label{subsec:partitions-learning}

In this section, we show a series of lemmas that imply that learnability of the hypothesis class \(\hclass\) is enough to imply that the minimum loss \(h \in \rc(\hclass)\) is learnable.  The proofs of these lemmas are largely straightforward sample complexity arguments, which we include for completeness.

To characterize the learnability of hypothesis classes, we find it simplest to go through \(\psib\)-dimension, as defined by \cite{BCL92}.

\begin{restatable}[\(\psib\)-dimension of \(\rc(\hclass)\)]{lemma}{psibdimensionofrchclass}
    For a hypothesis class \(\hclass\) of \(k\)-wise partitions over \(\instancespace\), the \(\psib\)-dimension of \(\rc(\hclass)\) can only be an order \(k\) factor (up to logarithmic factors) larger than the \(\psib\)-dimension of \(\hclass\).  That is, 
    \[\psib(\rc(\hclass)) \in O\bigg( k \psib(\hclass) \log (k \psib (\hclass)) \bigg).\]
    \label{lem:psi-b-dimension-of-rc-hclass}
\end{restatable}

\begin{proof}
    Suppose that \(\rc(\hclass)\) can \(\psib\) shatter \(n\) points \(x_1, \dots, x_n \in \instancespace\).  This means that there exist functions \(\psi_1, \dots, \psi_n : [k] \rightarrow \{0, 1\}\) such that for every labeling \(y \in \{0, 1\}^n\), there exists an \(h' \in \rc(\hclass)\) such that 
    \[\psi_i (h' (x_i)) = y_i, \qquad \forall i.\]
    Let \(S\) be a set containing one \(h'\) achieving each labeling \(y\).  Thus, \(|S| = 2^n\). Associate each \(h' \in S\) with a choice of \(\varphi \) and \(h \in \hclass\) such that \(h' = \varphi \circ h\).  

    There are at most \(k^k\) possible values of \(\varphi\).  Thus, there must be some \(\varphi^* : [k] \rightarrow [k]\) such that at least \(2^n / k^k\) elements of \(S\) are associated with \(\varphi^*\).  Let \(S'\) be the subset of \(S\) containing \(h'\) that are associated with \(\varphi^*\).  We have that each \(h' \in S'\) maps to a distinct value of 
    \[(\psi_1 (h' (x_1)), \dots, \psi_n (h' (x_n))) = ( (\psi_1 \circ \varphi^*) (h (x_1)), \dots, (\psi_n \circ \varphi^*) (h (x_n))).\]
    Let \(\hclass'\) be the set of \(h\) that are associated with \(h' \in S'\).  The above tells us that for functions \((\psi_1 \circ \varphi^*), \dots, (\psi_n \circ \varphi^*)\), the hypothesis classes in \(\hclass'\) span at least \(2^n / k^k = n^{\frac{n - k \log k}{\log n}}\) labelings of \(x_1, \dots, x_n\).  

    By the Sauer-Shelah lemma, this means that there must be a subset \(X' \subseteq (x_1, \dots, x_n)\) of size \(\Omega((n - k \log k) / \log n)\) that is \(\psib\)-shattered by \(S'\) with respect to the functions \((\psi_1 \circ \varphi), \dots, (\psi_n \circ \varphi)\).  

    Let \(q\) be the \(\psib\)-dimension of \(\hclass\).  
    Since \(q \in \Omega(\Omega((n - k \log k) / \log n))\), we have that \(n \in O(kq \log kq)\).
    Since this applies when \(n\) is \(\psib\)-dimension of \(\rc(\hclass)\), we have that 
    \[\psib(\rc(\hclass)) \in O\bigg( k \psib(\hclass) \log (k \psib (\hclass)) \bigg).\]
\end{proof}

A simple reduction can construct an ERM oracle for \(\rc(\hclass)\) using calls to an ERM oracle for \(\hclass\), where the number of calls depends only on \(k\).  However, we note that it may be possible to get a much more efficient ERM oracle, for example when \(\hclass\) is rotationally complete to begin with.   

\begin{restatable}[Converting ERM for \(\hclass\) to ERM for \(\rc(\hclass)\)]{lemma}{convertingermforhclasstoermforrchclass}
    For a hypothesis class \(\hclass\) of \(k\)-wise partitions of \(\instancespace\) and an empirical distribution \(\distribution\) over \(\instancespace \times \solutionspace\), given access to an oracle \(\oracle\) that can answer queries of the form 
    \[\argmin_{h \in \hclass} \ell_\centers(h, \distribution)\]
    for any set of centers \(\centers \in \solutionspace^k\), 
    we can construct an oracle \(\oracle_\rc\) that can answer queries of the form 
    \[\argmin_{h' \in \rc(\hclass)} \ell_\centers(h', \distribution)\]
    for any set of centers \(\centers \in \solutionspace^k\) using \(k^k\) calls to \(\oracle\).
\end{restatable}

\begin{proof}
    There are exactly \(k^k\) possible choices for \(\varphi: [k] \rightarrow [k]\).  Define 
    \[\varphi(\centers) = (\centers^{(\varphi(1))}, \dots, \centers^{(\varphi(k)})).\]
    Then we have that 
    \[\ell_\centers(\varphi \circ h, \distribution) = \ell_{\varphi(\centers)}(h, \distribution).\]
    Therefore, 
    \begin{align*}
        \argmin_{h' \in \rc(\hclass)} \ell_{\centers}(h', \distribution) &= \min_{\varphi: [k] \rightarrow [k]} \ell_\centers \left( \varphi \circ \left(\argmin_{h \in \hclass} \ell_{\centers}(\varphi \circ h, \distribution) \right), \distribution \right) \\
        &= \min_{\varphi: [k] \rightarrow [k]} \ell_\centers \left( \varphi \circ \left(\argmin_{h \in \hclass} \ell_{\varphi(\centers)}(h, \distribution)  \right) , \distribution \right).
    \end{align*}  
    Thus, we can evaluate this by enumerating over all choices of \(\varphi\), and making one call to \(\oracle\) for each \(\varphi\).  Then, we can evaluate the empirical loss of \(\varphi\) with the minimizing \(h \in \hclass\), and choose the \(\varphi\) that achieves the lowest loss.  
\end{proof}

\begin{restatable}[Pseudo-dimension of loss functions]{lemma}{pseudodimensionoflossfunctions}
    For a fixed \(\centers\), and hypothesis class \(\hclass\) of \(k\)-wise partitions over \(\instancespace\), the pseudo-dimension of the set of loss functions \(\ell_\centers(h, \cdot )\) for \(h \in \hclass\) is bounded by the \(\psib\)-dimension of \(\hclass\).
    \label{lem:pseudo-dimension-of-loss-functions}
\end{restatable}

\begin{proof}
    Let \(d\) be the pseudo-dimension of the set of \(\ell_\centers(h, \cdot)\) for \(h \in \hclass\).  This means that there exist points \(x_1, \dots, x_d \in \instancespace\) and thresholds \(t_1, \dots, t_d \in \mathbb{R}\) such that for every labeling \(y \in \{0, 1\}^n\), there exists an \(h \in \hclass\) such that for all \(i \in [d]\)
    \[ \begin{cases}
        \ell_\centers(h, x_i) \le t_i & \text{if } y_i = 0\\
        \ell_\centers(h, x_i) > t_i & \text{if } y_i = 1
    \end{cases} .\]

    For each \(i\), we construct a function \(\psi_i : [k] \rightarrow \{0, 1\}\) via 
    \[\psi_i(z) = \begin{cases}
        0 & \text{if } \md(\centers^{(i)} , x_i) \le t_i \\
        1 & \text{if } \md(\centers^{(i)}, x_i) > t_i
    \end{cases}.\]
    This ensures that for a hypothesis \(h \in \hclass\), \(\psi_i (h(x_i)) = 0\) if and only if \(\ell_\centers(h, x_i) \le t_i\).  Thus, \(\psi_1, \dots, \psi_d\) witness the \(\psib\)-shattering of \(x_1, \dots, x_d\), and the pseudo-dimension of the set of \(\ell_\centers(h, \cdot)\) for \(h \in \hclass\) is at most the \(\psib\)-dimension of \(\hclass\).  
\end{proof}

\begin{lemma}[Learning guarantee]
    Given an arbitrary set of centers \(\centers\) over a bounded space \(\solutionspace\), a learnable (finite \(\psib\)-dimension) rotationally complete hypothesis class \(\hclass\) of \(k\)-wise partitions of \(\instancespace\), an ERM oracle \(\oracle\) that can compute \(\argmin_{h \in \hclass} \ell_\centers(h, \distribution')\) for empirical distributions \(\distribution'\), and sample access to a distribution \(\distribution\) over \(\instancespace \times \solutionspace\), we can learn an \(h \in \hclass\) and a set of centers \(\centers_h \in \solutionspace^k\) such that 
    \[\E_{(I, S) \sim \distribution} \left[ \md(\centers_h^{h(I)}, S) \right] \le O(1) \left( 1 + \cost(\centers, \distribution) + \min_{h' \in \hclass} \cost(h', \distribution) \right).\]
    In particular, if \(d\) is the \(\psib\)-dimension of \(\hclass\), and \(\dmax\) is the largest distance in \(\solutionspace\), with probability \(1 - \delta\) we can learn the above in 
    \[O(\dmax^2 (d + \ln \frac{1}{\delta}))\]
    samples, where \(\dmax\) is the largest distance in \(\solutionspace\).
    \label{lem:learning-guarantee-for-k-wise-partitions}
\end{lemma}

\begin{proof}
    Since \(\hclass\) has \(\psib\) dimension \(d\), by \Cref{lem:pseudo-dimension-of-loss-functions} the set of loss functions \(\mathcal{L} = \{\ell_\centers(h, \cdot) \mid h \in \hclass\}\) has pseudo-dimension at most \(d\).  

    By uniform convergence bounds for \(\psib\)-dimension \cite{BCL92}, we have that for any \(\delta \in (0, 1)\), any distribution \(\distribution\) over \(\instancespace \times \solutionspace\), since all of the loss functions in \(\mathcal{L}\) have value bounded in \([0, D]\), \(m = O(\dmax^2 (d + \ln \frac{1}{\delta}))\) samples are sufficient to ensure that with probability \(1 - \frac{\delta}{4}\) over the draw of \(\{(I, S)_1, \dots, (I, S)_m \} \sim \distribution^m\), for all \(h \in \hclass\), the difference between the average empirical loss over the samples and the expected loss over \(\distribution\) is at most a constant, i.e.: 
    \begin{equation} 
    \left| \frac{1}{m} \sum_{j = 1}^m \ell_\centers(h, (I, S)_j) - \ell_\centers(h, \distribution) \right| \le 1.
    \label{eq:uniform-convergence}
    \end{equation}

    Let \(\distribution_m\) be the empirical distribution of \(m\) samples drawn i.i.d. from \(\distribution\).  The ERM oracle \(\oracle\) then gives us \(\widehat{h}\) such that 
    \[\widehat{h} = \argmin_{h \in \hclass} \ell_\centers(h, \distribution_m).\]

    \Cref{lem:ERM-approximation} implies that, for a rotationally complete \(\hclass\), and
    \[h^* = \argmin_{h \in \hclass} \cost(h, \distribution),\]
    the ERM minimizer \(\widehat{h}\) achieves
    \[\ell_{\centers}(\widehat{h}, \distribution_m) \le O(1) \cdot (\cost(h^*, \distribution_m) + \cost(\centers, \distribution_m)),\]
    Combining this with \Cref{eq:uniform-convergence}, we get that 
    \[\ell_{\centers}(\widehat{h}, \distribution) \le O(1) \cdot (1 + \cost(h^*, \distribution_m) + \cost(\centers, \distribution_m)).\]

    We bound the deviation of \(\cost(h^*, \distribution_m)\) and \(\cost(\centers, \distribution_m)\) from \(\cost(h^*, \distribution)\) and \(\cost(\centers, \distribution)\).  Specifically, let \(\centers_{h^*}\) be the optimal set of centers for each partition of \(h^*\) with respect to \(\distribution\).  We can write \(\cost(h^*, \distribution_m) \le \frac{1}{m} \sum_{i = 1}^m X_i\), where \(X_i = \md \left(C_{h^*}^{(h^*(I_i))}, S_i \right) \), where \((I_i, S_i)\) are the independent samples drawn from \(\distribution\).  Let \(X = \sum_{i = 1}^m X_i\).  By a Chernoff-Hoeffding bound, we have that   
    \begin{align*}
        \Pr \left[ X \ge \E[X] + m \right] &\le \exp \left( - \frac{(\E[X] + m)^2}{\E[X] + m} \right) \\
        &\le \exp( - \frac{1}{2}(\E[X] + m)).
    \end{align*}
    Thus, we have that as long as \(m \ge O(\log (\frac{1}{\delta}))\), we have 
    \[\cost(h^*, \distribution_m) \le 2 \cdot \cost(h^*, \distribution) + 1,\]
    with probability \(\ge 1 - \frac{\delta}{4}.\)

    Similarly, let \(Y = Y_1 + \dots + Y_m\), where \(Y_i = \md(\centers(S_i), S_i)\) where the \(S_i\) are the \(m\) independent samples we drew from \(\distribution\).  Then \(\cost(\centers, \distribution_m) = \frac{1}{m} \sum_{i = 1}^m Y_i\).  By a Chernoff-Hoeffding bound, we have 
    \begin{align*}
        \Pr \left[ Y \ge \E[Y] + m \right] &\le \exp( - \frac{1}{2}(\E[Y] + m)).
    \end{align*}
    Thus, as long as \(m \ge O(\log (\frac{1}{\delta}))\), we have 
    \[\cost(\centers, \distribution_m) \le \cost(\centers, \distribution) + 1,\]
    with probability \(\ge 1 - \frac{\delta}{4}\).

    Taking a union bound over these events, we have that with \(m = O(\dmax^2 (d + \ln \frac{1}{\delta}))\) samples, 
    \[\E_{(I, S) \sim \distribution} \left[ \md(\centers^{\widehat{h}(I)}, S) \right] \le O(1) \left( 1 + \cost(\centers, \distribution) + \min_{h' \in \hclass} \cost(h', \distribution) \right),\]
    with probability \(\ge 1 - \delta\).
\end{proof}

\subsection{Algorithmic guarantee}
\label{subsec:partitions-guarantee}

Finally, we can use the above procedure to design an algorithm that can take advantage of course information in the form of partitions of \(\solutionspace\), to achieve a potentially improved runtime guarantee for future instances. 

\begin{theorem}[Competing with a hypothesis class of \(k\)-wise partitions offline]
    Given an algorithm-with-predictions \(\alg\), a learnable (finite \(\psib\)-dimension) rotationally complete hypothesis class \(\hclass\) of \(k\)-wise partitions of \(\instancespace\), an ERM oracle \(\oracle\) that can compute \(\argmin_{h \in \hclass} \ell_\centers(h, \distribution')\) for empirical distributions \(\distribution'\), and \(O(\dmax^2 (d + \ln \frac{1}{\delta}))\) samples from distribution \(\distribution\) over \(\instancespace \times \solutionspace\), 
    it is possible to learn an \(h : \instancespace \rightarrow [k]\) and set of centers \(\centers = \left( C^{(1)}, \dots, C^{(k)} \right)\) that are an \(O(1)\) approximation to 
    \[\min_{h \in \hclass} \min_{\centers \in \solutionspace^k} \E_{(I, S) \sim \distribution} \left[ \md(S, C^{(h(I))}) \right].\]

    Furthermore, given an \(h : \instancespace \rightarrow [k]\) and \(\centers = \left( C^{(1)}, \dots, C^{(k)} \right)\), we can design a procedure such that the expected time to solve a new instance is 
    \[\E_{(I, S) \sim \distribution} \left[ \mathrm{runtime}(h(I)) + \md(S, C^{(h(I))})\right].\]
    \label{thm:competing-with-hypothesis-class}
\end{theorem}

\begin{proof}
    By \Cref{lem:learning-k-fixed-points} we know that we can use samples from \(\distribution\) to find a set of centers \(\widehat{\centers}\) that are an \(O(1)\)-approximation to 
    \[\argmin_\centers \cost(\centers, \distribution).\]
    By \Cref{lem:learning-guarantee-for-k-wise-partitions} we have that, using this \(\widehat{\centers}\), and additional samples from \(\distribution\), we can find a hypothesis \(\widehat{h} \in \hclass\) such that 
    \begin{align}
        \E_{(I, S) \sim \distribution} \left[ \md(\widehat{\centers}^{\widehat{h}(I)}, S) \right] &\le O(1) \left( 1 + \cost(\widehat{\centers}, \distribution) + \min_{h' \in \hclass} \cost(h', \distribution) \right). \nonumber \\
        &\le O(1) \left( 1 + \min_{\centers \in \solutionspace^k} \cost(\centers, \distribution) + \min_{h' \in \hclass} \cost(h', \distribution) \right) \nonumber \\
        &\le O(1) \cdot \min_{h' \in \hclass} \cost(h', \distribution) \label{eq:hypothesis-more-costly-than-clustering} \\
        &\le O(1) \cdot \min_{h \in \hclass} \min_{\centers \in \solutionspace^k} \E_{(I, S) \sim \distribution} \left[ \md(S, C^{(h(I))}) \right], \nonumber
    \end{align}
    where the \Cref{eq:hypothesis-more-costly-than-clustering} follows because the cost of the best hypothesis is lower bounded by the cost of the best clustering, and because we assume all distances in \(\solutionspace\) are lower bounded by 1 (see preliminaries, \Cref{sec:preliminaries}).  

    Finally, once we have a hypothesis \(h \in \hclass\) and the centers \(\centers_h\), on subsequent instances \(I\) drawn from \(\distribution\), we can compute \(\alg(I, \centers_h^{(h(I))})\), in expected time 
    \[\E_{(I, S) \sim \distribution} \left[ \mathrm{runtime}(h(I)) + \md(S, C^{(h(I))})\right].\]
\end{proof}

\section{Competing with trajectories online}

In this section, we consider the problem of solving a series of instance that arrive online. 
Consider an instance \(I \in \instancespace\) with (unknown) true solution \(S \in \solutionspace\).  As before, we assume that the solution space \(\solutionspace\) is equipped with a metric distance \(\md\), which we can calculate efficiently (see Assumption \ref{ass:metric-distance-computable}).  We can think of an algorithm with predictions as searching the solution space \(\solutionspace\) by growing a ball around a prediction \(P \in \solutionspace\) of our choice.  The cost of the algorithm is the radius of the ball necessary to find \(S\).  Furthermore, we can search from multiple points, and the total cost will be the sum of the radii searched from each of the points.  

\begin{definition}[Online Ball Search]\label{def:online-ball-search}
    In the \emph{online} setting, we consider instances arriving over \(T\) days.  On each day \(t \in [T]\) an instance \(I_t \in \instancespace\) arrives, with (unknown) solution \(S_t \in \solutionspace\).   We are given access to an algorithm with predictions \(\alg\) such that for any prediction \(P \in \mathbf{S}\), the runtime of \(\alg(I_t, P)\) is bounded by \(\md(S_t, P)\).  On each day \(t \in [T]\), the algorithm must output \(S_t\).  
\end{definition}

In the online setting, we want to minimize the total work that the algorithm does over a sequence of \(T\) instances.   We make \emph{no distributional assumptions} about the instance-solution pairs that arrive.  Instead, we approach the problem through the lens of \emph{competitive analysis}, and provide algorithms that compete with the best strategies in hindsight.  

In \Cref{subsec:offline-baselines}, we define the offline strategies that we are competing against.  This requires careful consideration of various modeling constraints.  In \Cref{subsec:reduction-to-k-server}, we give an algorithm that is \(O(k^2)\)-competitive against these baselines in the total radius searched, via a reduction to \(k\)-server.  We note that this does not immediately solve our problem of interest, as we wish to bound the \emph{total work} done by our algorithm, which is only lower bounded by the total radius searched.  In \Cref{subsec:improved-runtime}, we give an algorithm that is \(O(k^4 \ln^2 k)\)-competitive in the total radius searched, and which has total runtime bounded by an \(O(1)\) factor times the total radius searched.  To approach this, we use significantly different techniques than in the reduction to \(k\)-server.

\subsection{Offline Baselines}
\label{subsec:offline-baselines}
We will analyze algorithms for this problem under the paradigm of competitive analysis, where we compare the performance of our algorithm to the performance of the best offline strategy in hindsight.  To do this, we must define the set of offline strategies that we are competing with.  

In our problem setting, on each day, the algorithm is allowed to search outward from any set of points, and pays for the total radius that is searched before that day's point is found.  A natural idea is to allow the offline strategies to have the same form.  However, the offline strategy has full knowledge of all of the requests that arrive in the sequence.  So, on each day, it could search exactly from that day's request, achieving zero cost.  Thus, this set of offline strategies is not particularly meaningful.  

This motivates us to consider a restricted set of offline strategies, that captures some structure in the input that we could hope to take advantage of.  This is reasonable, as we must leverage some structure to get a fast algorithm for sequences of inputs.  Otherwise, if we had a fast algorithm for arbitrary sequences of inputs, this would imply a faster algorithm for one arbitrary input.  

Previous work imposes structure by providing a guarantee that competes with the best fixed prediction in hindsight \cite{KBTV22}.  This can be interpreted as imposing the structure that the data is well-clustered, and therefore the best fixed point performs well for the distribution of inputs.  This objective function is analogous to a 1-medians objective, 
as our cost is the sum of the distances from the requests to the center (best fixed prediction in hindsight).

The first extension is to allow the offline strategy to move over time.  We must be careful about how we do this, as allowing the offline point to move arbitrarily on each day results in the zero cost issue that we discussed earlier.  Thus, we allow to the point to move, but we charge the offline strategy for the total distance that the point moves.  That is, the cost to the offline strategy on a day \(t\) is the sum of the \emph{movement cost}, the distance that the strategy moved its prediction on day \(t\), and the \emph{hit cost}, the distance from the request to the prediction after the prediction moves.  We call this strategy a \emph{trajectory}.

\begin{definition}[Trajectory]
    A \emph{trajectory}, with respect to a sequence of instances \(\{I_t : t \in [T]\}\), is a sequence of predictions \(P_t \in \mathbf{S}\), for \(t \in [T]\).  The \emph{cost} of a trajectory is the sum of the prediction costs for each day's prediction (``hit cost"), plus the total distance that the trajectory moves (``movement cost"),  
    \[\mathrm{cost}(\{P_t : t \in [T]\}) = 
    \left[ \sum_{i = 1}^T \md(P_t, S_t) \right] + \left[ \sum_{i = 1}^T \md(P_{t - 1}, P_t) \right],\]
    where \(S_t\) is the solution associated with instance \(I_t\), and we define the initial \(P_0 = \mathbf{0}\), is a fixed arbitrary starting prediction.
\end{definition}

The second extension is that we compete against offline strategies that consist of \(k\) points.  That is, on each day, the hit cost of the offline strategy is the distance from the request to the nearest of the \(k\) points.  This setting is analogous to a \(k\)-medians objective, as the cost of each request is the distance from the request to the closest of \(k\) centers.  We note that the extension from one center to \(k\) centers makes the offline problem of choosing the best strategy go from being convex, to being non-convex.  

\begin{definition}[Multiple trajectories]
    A collection of \emph{\(k\) trajectories}, with respect to a sequence of instances \(\{I_t : t \in [T]\}\), is a sequence of predictions \(\{P_t \in \solutionspace: t \in [T]\}\), and associates each day \(t\) with one of the trajectories \(\in [k]\).  We use \(\mathcal{T}^{(i)} \subseteq [T]\) to refer the days that are associated with trajectory \(i\), and we denote the collection of trajectories by \(\left(\{P_t : t \in \mathcal{T}^{(1)}\}, \dots, \{P_t : t \in \mathcal{T}^{(k)}\} \right)\).

    For a day \(t \in [T]\), we use \(\prev^{(i)}(t)\) to refer to the closest day previous to \(t\) that belongs to trajectory \(i\), and \(\mnext^{(i)}(t)\) to refer to the closest day subsequent to \(t\) that belongs to trajectory \(i\).  If \(t\) is the first day in its trajectory, then we define \(\prev^{(i)}(t) = P_0 = \mathbf{0}\).  

    The \emph{cost} of the collection of trajectories is the sum of the costs of the \(k\) trajectories.  The cost of trajectory \(i\) is the single trajectory cost of \(\{P_t : t \in \mathcal{T}^{(i)}\}\) with respect to \(\{I_t : t \in \mathcal{T}^{(i)}\}\).  That is, 
    \[\cost(\{P_t : t \in \mathcal{T}^{(i)}\}) = \left[ \sum_{t \in \mathcal{T}^{(i)}} 
    \md(P_t, S_t) \right] + \left[ \sum_{t \in \mathcal{T}^{(i)}} 
    \md(P_{\prev^{(i)}(t)}, P_t) \right] ,\]
    \[\cost\left(\{P_t : t \in \mathcal{T}^{(1)}\}, \dots, \{P_t : t \in \mathcal{T}^{(k)}\} \right) = \sum_{i = 1}^k \cost(\{P_t : t \in \mathcal{T}^{(i)}\}) .\]
\end{definition}

We note that while the offline strategy is limited to maintaining a palette of \(k\) predictions, and must pay to move the predictions, we do not make this requirement of the online algorithm.   

Another way to interpret a \(k\)-trajectory baseline for a sequence of requests, is that it is essentially the offline \(k\)-server cost for those requests.  In the \(k\)-server problem, the algorithm must maintain a set of \(k\)-servers in a metric space \(\mathbf{S}\).  On each day, a request arrives at some point \(S \in \mathbf{S}\), and the algorithm must move one of its servers to \(S\).  The cost to the algorithm is the total distance it moves its servers.  The offline baseline is the best way to move the servers, when the sequence of requests is known ahead of time.

\begin{lemma}[Multiple trajectories are approximately \(k\)-server baselines] \label{lem:trajectories-equal-offline-k-server}
    Let \(\mathcal{S} = (S_1, \dots, S_T)\) be a sequence of \(T\) requests (solutions).  Let \(\{P_t : t \in \mathcal{T}^{(1)}\}, \dots, \{P_t : t \in \mathcal{T}^{(k)}\}\) be the offline optimal collection of \(k\) trajectories for \(R\).  Let \(\mathrm{serveropt}_k(\mathcal{S})\) be the cost of the offline optimal \(k\)-server solution serving \(\mathcal{S}\).  We have that 
    \[\cost(\{P_t : t \in \mathcal{T}^{(1)}\}, \dots, \{P_t : t \in \mathcal{T}^{(k)}\}) \le \mathrm{serveropt}_k(\mathcal{S}) \le 2 \cdot \cost(\{P_t : t \in \mathcal{T}^{(1)}\}, \dots, \{P_t : t \in \mathcal{T}^{(k)}\}).\]
\end{lemma}

\begin{proof}
    The first inequality follows from observing that a \(k\)-server solution that serves the requests \(\mathcal{S}\) is a valid collection of \(k\) trajectories, that has movement cost equal to the cost of the \(k\)-server solution, and zero hit cost.  

    For the second inequality, we convert the collection of \(k\) trajectories into an offline \(k\)-server solution as follows.  We assign each of trajectory one server, that services the requests that correspond to that trajectory.  We can bound the cost to server \(i\) by the cost to trajectory \(i\).  We can bound the cost of this \(k\)-server solution by 
    \begin{align*}
        \sum_{t \in \mathcal{T}^{(i)}} \md(S_{\mathrm{prev}^{(i)}(t)}, S_t) &\le \sum_{t \in \mathcal{T}^{(i)}} \left[ \md(S_{\mathrm{prev}^{(i)}(t)}, P_{\mathrm{prev}^{(i)}(t)}) + \md(P_{\mathrm{prev}^{(i)}(t)}, P_t) + \md(P_t, S_t) \right] \\
        &\le 2 \left[ \sum_{t \in \mathcal{T}^{(i)}} \md(S_{t}, P_t) \right] +  \left[ \sum_{t \in \mathcal{T}^{(i)}} \md(P_{\mathrm{prev}^{(i)}(t)}, P_t) \right] \\
        &\le 2 \cdot \cost(\{P_t : t \in \mathcal{T}^{(1)}\}, \dots, \{P_t : t \in \mathcal{T}^{(k)}\}).
    \end{align*}
\end{proof}

This observation already gives the following algorithm that is competitive against one trajectory.  

\begin{algorithm}
\caption{``Predict yesterday's solution"}
\label{alg:predict-yesterdays-solution}
\begin{algorithmic}[1]
    \STATE \(P = \mathbf{0}\) 
    \FOR{day \(t\), instance \(I_t\) arrives}
        \STATE \(S_t = \alg(I_t, P)\)
        \STATE \(P = S_t\)
        \STATE Output solution \(S_t\)
    \ENDFOR
\end{algorithmic}
\end{algorithm}

\begin{corollary}[``Predict yesterday's solution" competes with any single trajectory]
    The ``predict yesterday's solution" strategy (\Cref{alg:predict-yesterdays-solution}) is \(O(1)\)-competitive with the best single trajectory in hindsight.  Furthermore, the total runtime of \Cref{alg:predict-yesterdays-solution} is bounded by \(O(1)\) times the total radius searched. 
    \label{cor:predict-yesterday-is-competitive}
\end{corollary}

\begin{proof}
    \Cref{lem:trajectories-equal-offline-k-server} tells us that the best offline trajectory has cost within a factor 2 of the optimal offline 1-server strategy.  For any sequence of requests, the optimal 1-server strategy is simply the one that moves the single server to each request on each day, and pays cost 
    \[\sum_{t = 1}^T \md(S_{t - 1}, S_t).\]
    This is also the cost of \Cref{alg:predict-yesterdays-solution}.  Thus, \Cref{alg:predict-yesterdays-solution} is 2-competitive against any fixed strategy.  
    
    Furthermore, the runtime of \Cref{alg:predict-yesterdays-solution} is dominated by the time to run the algorithms-with-predictions subroutine, so the total runtime of \Cref{alg:predict-yesterdays-solution} is bounded by \(O(1)\) times the total radius searched. 
\end{proof}

\begin{remark*}
    The guarantee given by this strategy is already stronger than what is shown in previous work.  \cite{KBTV22} give a strategy that competes with the best \emph{fixed} prediction in hindsight.  This strategy allows us to compete with trajectories that are adaptive, in that they move over time.   
\end{remark*}

\subsection{Reduction to \(k\)-server}
\label{subsec:reduction-to-k-server}

We consider the setting of competing against a collection of \(k\) trajectories. 
In this setting, we can think of the offline optimum as maintaining \(k\) prediction points.  On each day, the cost is the distance from today's solution to the minimum of these points.  Like the single trajectory setting, these points can move over time, and the offline cost is charged for the total distance that the \(k\) points move, capturing an adaptive strategy.  

In this section, we show how get a competitive algorithm for the online ball search problem using a competitive algorithm for the \(k\)-server problem, with an \(k\) factor blow-up in the competitive ratio.  We then discuss why this does not immediately imply a fast algorithm for solving sequences of instances. 

In the \(k\)-server problem, an algorithm must maintain the positions of \(k\) servers in a metric space \(\solutionspace\).  On each day \(t \in [T]\), a request \(r_t \in \solutionspace\) is made, and the algorithm must move at least one server to \(r_t\).  The cost to the algorithm is the total distance that all of the servers move.  The competitive ratio of the algorithm is given with respect to the best offline solution, i.e. the best way to service the requests in order with \(k\) servers in hindsight.  For an 
introduction to the \(k\)-server problem, the reader is referred to the survey \cite{Kou09}, with the note that there has been a good deal of work on this problem since the survey was published.  

\begin{theorem}[Online ball search to \(k\)-server reduction]
    Given an \(\alpha\)-competitive algorithm for the \(k\)-server problem on metric space \(\solutionspace\), we can construct a \(O(k \alpha)\)-competitive algorithm for the online ball search problem on \(\solutionspace\).  
    \label{thm:online-to-k-server-reduction}
\end{theorem}

\begin{proof}
    We maintain \(k\) predictions \(P_1, \dots, P_k \in \solutionspace\), which can also be thought of as the position of \(k\) servers.  

    On each day \(t \in [T]\), instance \(I_t\) arrives.  Our algorithm runs \(\alg(I_t, P_i)\) for all \(i \in [k]\) in parallel\footnote{Here, by ``in parallel" we refer to interleaving the steps of the various threads, resulting in a sequential algorithm.} at equal rates, until one of them terminates and finds \(S_t\).  Then, we take \(S_t\) to be the request for the \(k\)-server problem for day \(t\), and run our competitive \(k\)-server algorithm to update the positions of the predictions/servers.  

    The cost to the online ball search problem is \(k\) times the distance from the closest prediction (at the beginning of day \(t\)) to \(S_t\).  Since the \(k\)-server algorithm must service \(S_t\), the cost to the \(k\)-server algorithm on day \(t\) is at least the distance from the closest server/prediction (at the beginning of day \(t\)) to \(S_t\).  Thus the cost to the online ball search algorithm is bounded by at most \(k\) times the cost to the \(k\)-server algorithm.

    Finally, the \(k\)-server algorithm is \(\alpha\)-competitive against the best offline \(k\)-server solution that visits all of the requests \(S_t\).  By \Cref{lem:trajectories-equal-offline-k-server}, we know that this has cost at most \(O(1)\) times more than the \(k\) best trajectories for the sequence of \(S_t\)s.  Thus, this algorithm for online ball search is \(O(k \alpha)\)-competitive against \(k\) trajectories in the number of steps it takes of the subroutine \(\alg\).
\end{proof}

\begin{corollary}[\(O(k^2)\)-competitive algorithm for online ball search]
    There exists an \(O(k^2)\)-competitive algorithm for the online ball search problem.
    \label{cor:online-ball-search-via-k-server}
\end{corollary}

This follows from the fact that the work function algorithm is \(O(k)\)-competitive for the \(k\)-server problem \cite{KP95}.    

The strategy of reducing online ball search to the \(k\)-server problem is both promising and presents some challenges.  On the one hand, the \(k\)-server problem is a natural and very well-studied problem in online algorithms.  It is promising that improvements to algorithms for the \(k\)-server problem can naturally translate to the online ball search problem.  

On the other hand, the online ball search problem appears to have some advantages that are lost in the reduction to \(k\)-server.  For example, an online ball search algorithm is not limited to maintaining a palette of exactly \(k\) predictions.  This allows us in \Cref{subsec:improved-runtime} to design an algorithm that achieves competitive guarantees for all values of \(k\) simultaneously.  Another issue with a \(k\)-server strategy, is that existing algorithms for the \(k\) server problem can be computationally intensive.  The work function algorithm, for example, requires solving a max-flow instance on each day \(t\) with \(O(t)\) vertices.  This causes the work to scale badly as the number of days grows large.  Finally, some algorithms for the \(k\)-server problem work well in the setting where the number of points \(n\) in the metric space \(\solutionspace\) is finite.  A notable example is the algorithm of Bansal, Buchbinder, Madry, and Naor which gives a randomized algorithm that achieves competitive ratio of \(\widetilde{O}(\log^2 k \log^3 n)\) against an oblivious adversary \cite{BBMN15}.  However, this is not ideal for our applications of interest, in which the solution space \(\solutionspace\) is often a space of vectors \(\subseteq \mathbb{R}^d\) with distance \(\md(u, v) = ||u - v||_q\) corresponding to a relevant norm \(q\).

\subsection{Algorithm with improved runtime}
\label{subsec:improved-runtime}
In this section we show an algorithm that is \(O(k^4 \ln^2 k)\)-competitive against any set of \(k\) trajectories.  Furthermore, the algorithm is deterministic (and thus resistant to an adaptive adversary), and is oblivious to the choice of \(k\).  That is, the competitive ratio holds for all \(k\) simultaneously.  Importantly, this algorithm also has total runtime bounded that is \(O(1)\) times the sum of radii searched.  We note that the techniques we use in this algorithm are significantly different than those in the reduction to \(k\)-server (\Cref{subsec:reduction-to-k-server}).

\Cref{alg:quadratic-decay} generalizes the single trajectory algorithm.  For the single trajectory case, on each day the algorithm searched from ``yesterday's solution."  For multiple trajectories, this is no longer enough, as the previous day's solution could come from some other trajectory, and be arbitrarily far away from today's solution.  

\Cref{alg:quadratic-decay} addresses this by searching from \emph{all} previous solutions in parallel.  These ``threads" are run at approximately harmonic rates, i.e. the thread of the \(i\)th most recent solution is run at rate \(\frac{1}{i^2 \ln^2 i}\).  This leaves the issue that the previous solution from the same trajectory as today could be arbitrarily far in the past, and be run at an extremely slow rate.  This is addressed by ``pruning" threads that are no longer fruitful.  That is, when the ball searched around a solution \(i\) fully contains the ball searched around a previous solution \(j\), the algorithm can stop running thread \(j\), as that work is redundant with thread \(i\).  In this case, we say thread \(i\) ``subsumes" thread \(j\).  Then, we can increase the rates of the threads that are lower priority than \(j\), and maintain that there is at most one thread being run at each rate \(\frac{1}{i^2 \ln^2 i}\) for each integer \(i\).  In the analysis we show that either the algorithm runs long enough for the previous solution from the same trajectory as today to be elevated to rate \(\ge \frac{1}{k^2 \ln^2 k}\) and eventually solve the instance, or the algorithm terminates more quickly than that, which is even better.  This allows us to bound the competitive ratio by \(O(k^4 \ln^2 k)\).

The main runtime overhead in \Cref{alg:quadratic-decay} comes from checking when to prune the slower threads.  For each step that the algorithm takes of a thread, it must check whether the thread has been subsumed.  For a thread running at rate \(\frac{1}{i^2 \ln^2 i}\) it must check the \(O(i)\) threads that are faster than it on each step.  Thus, even though the thread is running steps of the subroutine \(\alg\) at rate \(\frac{1}{i^2 \ln^2 i}\), it is doing total work at rate \(O(i) \cdot \frac{1}{i^2 \ln^2 i} = O(\frac{1}{i \ln^2 i} )\). The rates are chosen so that this series converges, and the total work of the algorithm can be bounded by \(O(1)\) times the work it does for the fastest thread.  

\begin{algorithm}
\caption{Quadratic decay}
\label{alg:quadratic-decay}
\begin{algorithmic}[1]
    \FOR{day \(t\), instance \(I_t\) arrives}
        \STATE 
        \STATE Initialize linked-list of active threads
        \FOR{\(t' : t \ge t' \ge 1\)}
            \STATE Create thread \(t'\) to run \(\alg(I_t, S_{t'})\) 
            \STATE \(\mathrm{radius}(t') := 0\)
            \STATE Append thread \(t'\) to list of active threads
        \ENDFOR
        \STATE 
        \STATE Run active threads in parallel, 
        \STATE \quad with \(i\)th active thread in list running at rate \(\frac{1}{i^2 \ln^2 i}\) (rate 1 for \(i = 1\)) \label{algline:run-threads-at-different-rates-quadratic}
        \STATE \quad track radius of each thread
        \WHILE{no thread completed}
            \FOR{step taken of thread running at \(i\)th rate}
                \STATE \(t_1\) := thread associated with this rate, found by walking down linked list of active threads 
                \FOR{active thread \(t_2\) with faster rate than \(t_1\)}
                \label{algline:kill-check-quadratic}
                    \IF{\(\md(S_{t_1}, S_{t_2}) \le \radius(t_2) - \radius(t_1)\)} \label{algline:kill-condition-quadratic}
                        \STATE Kill thread \(t_1\)
                        \STATE Remove thread \(t_1\) from linked list of active threads
                    \ENDIF
                \ENDFOR
            \ENDFOR 
        \ENDWHILE
        \STATE 
        \STATE \(S_t = \) solution of thread that completed
        \STATE Output \(S_t\)
    \ENDFOR
\end{algorithmic}
\end{algorithm}

\begin{definition}[Subsumes]
    If thread \(i\) kills thread \(j\) (\Cref{alg:quadratic-decay}, line \ref{algline:kill-condition-quadratic}), then we say that thread \(i\) \emph{subsumes} thread \(j\).  
    If thread \(j\) previously subsumed thread \(h\), then thread \(i\) now also subsumes thread \(h\).
\end{definition}

\begin{definition}[Radius of a thread]
    For a thread \(i\) at some given time, \(\radius(i)\) is the number of steps that thread \(i\) has been run so far.  
    If thread \(i\) subsumes thread \(j\), for analysis we will set \(\mathrm{rate}(j) = \mathrm{rate}(i)\).  We also continue updating the radius of \(j\) i.e., \(\radius(j)\) continues to increase at \(\mathrm{rate}(j)\). 
\end{definition}

\begin{definition}[Virtual radius]
    The \emph{virtual radius} of \Cref{alg:quadratic-decay} is the radius of the highest rate thread.  
\end{definition}

We will show that the virtual radius of the algorithm is within \(O(1)\) both of the total radius searched (\Cref{lem:quadratic-virtual-radius-bounds-total-radius}) and the total runtime \(O(1)\) (\Cref{lem:quadratic-total-radius-bounds-runtime}), making it a useful abstraction.  

We define the subsuming condition so that thread \(i\) subsumes thread \(j\), exactly when the ball of radius \(\radius(i)\) around \(S_i\) fully contains the ball of radius \(\radius(j)\) around \(S_j\).  Thus the part of \(\solutionspace\) that has been searched by thread \(j\), has also been searched by thread \(i\), and the work of thread \(j\) is redundant.  The following lemma shows that this invariant is maintained over the run of the algorithm.  

\begin{restatable}[Subsuming identity]{lemma}{subsuminginequality}
    If thread \(i\) subsumes thread \(j\), then 
    \[\md(S_i, S_j) \le \radius(i) - \radius(j).\]
    \label{lem:subsumes-ineq}
\end{restatable}

\begin{proof}
    We show this via induction on the number of times the ``kill condition" (Line \ref{algline:kill-condition-quadratic}) is triggered in \Cref{alg:quadratic-decay}.  In the base case, at the outset of the algorithm, no thread subsumes any other thread, so the statement holds vacuously. 

    Now, consider the \(s\)th kill event.  Consider a thread \(i\) that subsumes a thread \(j\) at this point (not necessarily the threads that just triggered the kill condition).  
    There are two cases.  
    \begin{itemize}
        \item[\textit{Case 1.}] Thread \(i\) subsumed thread \(j\) previously, at the \((s - 1)\)th kill event.  This means that \(\rate(j) = \rate(i)\), so \(\radius(i)\) and \(\radius(j)\) increased by the same amount since the previous kill event.  Thus, 
        \[\md(S_i, S_j) \le \radius(i) - \radius(j)\]
        continues to hold from the induction hypothesis. 
        
        \item[\textit{Case 2.}] Thread \(i\) killed thread \(j\), or thread \(i\) killed a thread \(h\) that subsumes thread \(j\), at the \(s\)th kill event.  

        If thread \(i\) killed thread \(j\), then by the kill condition (\Cref{alg:quadratic-decay}, line \ref{algline:kill-condition-quadratic}), we have that 
        \[\md(S_i, S_j) \le \radius(i) - \radius(j).\]

        If thread \(i\) killed a thread \(h\) that subsumes thread \(j\), then the kill condition gives us that 
        \[\md(S_i, S_h) \le \radius(i) - \radius(h).\]
        Since thread \(h\) already subsumes thread \(j\), Case 1 gives us that 
        \[\md(S_h, S_j) \le \radius(h) - \radius(j).\]
        Therefore, we have 
        \begin{align*}
            \md(S_i, S_j) &\le \md(S_i, S_h) + \md(S_h, S_j) &\text{triangle inequality} \\
            &\le \radius(i) - \radius(h) + \radius(h) - \radius(j) \\
            &\le \radius(i) - \radius(j).
        \end{align*}
    \end{itemize}
\end{proof}

The previous lemma implies that if a thread \(j\) is subsumed by a thread \(i\), then thread \(i\) is at least as effective at finding that day's solution as thread \(j\).  

\begin{restatable}{lemma}{subsumedcompleteimpliessomeonecomplete}
    If a subsumed thread \(i\) would have completed, i.e., 
    \[\md(S_i, S_t) \le \radius(i),\]
    then there must be some active thread of \Cref{alg:quadratic-decay} that has completed.  
    \label{lem:subsume-completion}
\end{restatable}

\begin{proof}
    Let \(i^*\) be the thread that subsumes thread \(i\).  We have that 
    \begin{align*}
        \md(S_{i^*}, S_t) &\le \md(S_{i^*}, S_i) + \md(S_i, S_t) &\text{triangle inequality} \\
        &\le \md(S_{i^*}, S_i) + \radius(i) \\
        &\le \radius(i^*) & i^* \text{ subsumes } i, \Cref{lem:subsumes-ineq},
    \end{align*}
    and therefore, thread \(i^*\) completed.  
\end{proof}

Because we choose the rates of the threads to form a converging series, it is sufficient to bound the virtual radius of the algorithm, as it is within \(O(1)\) of the total radius searched by the algorithm.  

\begin{restatable}[Total radius in terms of virtual radius]{lemma}{workovertime}
    The sum of the radii searched by \Cref{alg:quadratic-decay} is at most \(O(1)\) times the virtual radius at the end of \Cref{alg:quadratic-decay}.  
    \label{lem:quadratic-virtual-radius-bounds-total-radius}
\end{restatable}

\begin{proof}
    Denote the virtual radius at the end of \Cref{alg:quadratic-decay} by \(v\).  The virtual radius is the radius of the highest rate thread (rate 1).  Thus, the highest rate thread contributes \(v\) to the total work.  

    Over the run of \Cref{alg:quadratic-decay}, at any given point, there is at most one thread per \(i \ge 2, i \in \mathbb{N}\) with rate equal to \(\frac{1}{i^2 \ln^2 i}\).  Consider the amount the thread(s) with rate \(\frac{1}{i^2 \ln^2 i}\) increase the total radius searched over time.  The ratio of the radius searched by this thread to the highest rate thread is \(\frac{1}{i^2 \ln^2 i}\).  So, over the time that the highest rate thread accrued radius \(v\), this thread has increased the total radius searched by \(\frac{v}{i^2 \ln^2 i}\).  

    Summing over all threads, we can bound the total radius searched by 
    \[v + \sum_{i = 2}^\infty \frac{v}{i^2 \ln^2 i} < 2 v = O(1) \cdot v.\]
\end{proof}

Now, we account for the overhead of running the algorithm.  We note that it is not only the total radius searched that contributes to the work done by the algorithm.  We must be careful to account for the work that it takes to check for when threads are subsumed.  

\begin{lemma}[Runtime in terms of total radius]
    On a given day \(t\), the runtime of \Cref{alg:quadratic-decay} is at most \(O(1)\) times the total radius searched at the end of \Cref{alg:quadratic-decay}.
    \label{lem:quadratic-total-radius-bounds-runtime}
\end{lemma}


\begin{proof}
    We assume that each operation of running a thread for ``one step" (\Cref{alg:quadratic-decay}, Line \ref{algline:run-threads-at-different-rates-quadratic}) can be done in \(O(1)\) time.  In particular we note that the interleaving schedule of the threads (based on rate) is fixed and does not depend on the input.  Thus it can be computed in advance.

    For each step that the \(i\)th fastest thread is run (\Cref{alg:quadratic-decay}, Line \ref{algline:run-threads-at-different-rates-quadratic}), the algorithm performs the following overhead: 
    \begin{itemize}
        \item The algorithm steps through the linked list to find the associated thread in \(O(i)\) time.
        \item The algorithm checks the thread against all faster threads to see if it has been subsumed for each of the \(i - 1\) faster threads.  This requires one distance query for each of the faster threads.  By Assumption \ref{ass:metric-distance-computable}, we have that this takes \(O(1)\) time.   
        \item If the thread has been subsumed, the algorithm removes it from the linked-list of active threads in \(O(1)\) time.
    \end{itemize}
    Thus, each step of this thread is accompanied by \(O(i)\) work.  

    Denote the virtual radius at the end of \Cref{alg:quadratic-decay} by \(v\).  Over the run of \Cref{alg:quadratic-decay}, the total work done by the thread(s) that run at the \(i\)th fastest rate is 
    \[v \cdot \frac{1}{i^2 \ln^2 i} \cdot O(i) = v \cdot O \left(\frac{1}{i \ln^2 i} \right).\]

    Summing over all threads, we have that the total work can be bounded by 
    \[v + \sum_{i = 2}^\infty v \cdot O \left(\frac{1}{i \ln^2 i} \right) = O(1) \cdot v.\]
    By \Cref{lem:quadratic-virtual-radius-bounds-total-radius}, we have that \(v\) is within \(O(1)\) of the sum of the radii searched by \Cref{alg:quadratic-decay}.  Thus, the total work done by \Cref{alg:quadratic-decay} is at most \(O(1)\) times the sum of the radii that it searches.
\end{proof}

A key part of our analysis is showing that even if all solutions close to today's solution \(S_t\) were seen very far in the past, our algorithm will eventually raise the rates of these threads to a reasonable amount.  We do this by showing that enough other threads will be subsumed, clearing the way for threads lower down the chain.  

\begin{lemma}[Slow threads have bounded lifetime]
    Let \(i\) be a thread that is running at rate \(\ge \frac{1}{(k - 1)^2 \ln^2 (k - 1)}\) in \Cref{alg:quadratic-decay} at some point. Let \(j\) be another thread running at a slower rate than \(i\).  Thread \(j\) is subsumed in at most \(k^3 \ln^2 k \cdot \md(S_i, S_j)\) additional units of virtual radius.   
    \label{lem:quadratic-subsumed-time-bound}
\end{lemma}

\begin{proof}
    Assume for the sake of contradiction that \(j\) is not subsumed in the next \(k^3 \ln^2 k \cdot \md(S_i, S_j)\) units of virtual radius.  For \(i\) faster than \(j\), and thread \(j\) not subsumed, \Cref{alg:quadratic-decay} maintains that \(\mathrm{rate}(i) > \mathrm{rate}(j)\).  Thus, \((\radius(i) - \radius(j))\) is always nonnegative, and monotonically nondecreasing with virtual radius. 
    Over the next \(k^3 \ln^2 k \cdot \md(S_i, S_j)\) units of virtual radius, we have that 
    \[\Delta (\radius(i) - \radius(j)) \ge k^3 \ln^2 k \cdot \md(S_i, S_j) \left(\mathrm{rate}(i) - \mathrm{rate}(j) \right),\]
    where \(\Delta (\radius(i) - \radius(j))\) is the change in \(\radius(i) - \radius(j)\).
    Since \(\mathrm{rate}(i) \ge \frac{1}{(k - 1)^2 \ln^2 (k - 1)}\) and \(i\) is faster than \(j\), we have that 
    \begin{align*}
        \mathrm{rate}(i) - \mathrm{rate}(j) &\ge \frac{1}{(k - 1)^2\ln^2 (k - 1)} - \frac{1}{k^2 \ln^2 k} \\
        &\ge \frac{1}{(k - 1)^2 \ln^2 k} - \frac{1}{k^2 \ln^2 k} \\
        &= \frac{2k + 1}{k^2 (k - 1)^2 \ln^2 k} \\
        &\ge \frac{1}{k^3 \ln^2 k}.
    \end{align*}
    
    Thus, after \(k^3 \ln^2 k \cdot \md(S_i, S_j)\) units of virtual radius, we have that 
    \begin{align*}
        \Delta (\radius(i) - \radius(j)) &\ge \md(S_i, S_j) \\
        \radius(i) - \radius(j) &\ge \md(S_i, S_j).
    \end{align*}
    If thread \(i\) is still active at the end of these iterations, this implies that thread \(i\) kills thread \(j\) (\Cref{alg:quadratic-decay}, line \ref{algline:kill-condition-quadratic}).  Otherwise let thread \(i^*\) be the active thread that subsumes thread \(i\).  By \Cref{lem:subsumes-ineq}, we have that 
    \[\md(S_i, S_{i^*}) \le \radius(i^*) - \radius(i).\]
    Thus, we get that 
    \begin{align*}
        \md(S_{i^*}, S_j) &\le \md(S_{i^*}, S_i) + \md(S_i, S_j) \\
        &\le \radius(i^*) - \radius(i) + \radius(i) - \radius(j) \\
        &\le \radius(i^*) - \radius(j),
    \end{align*}
    which implies that thread \(i^*\) will kill thread \(j\).  Thus, thread \(j\) must be subsumed, and we reach a contradiction.  
\end{proof}

\begin{theorem}[Competing with \(k\) trajectories online in radius \emph{and} runtime]
    \Cref{alg:quadratic-decay} is \(O(k^4 \ln^2 k)\)-competitive with any set of \(k\) trajectories in the total radius that it searches.  Furthermore, the total runtime of the algorithm over \(T\) days can be bounded by \(O(1)\) times the total radius that it searches.
    \label{thm:online-competitive-with-runtime}
\end{theorem}

\begin{proof}
    Fix a collection of trajectories \(\left(\{P_t : t \in \mathcal{T}^{(1)}\}, \dots, \{P_t : t \in \mathcal{T}^{(k)}\} \right)\).  We begin by bounding the cost of \Cref{alg:quadratic-decay} on days that correspond to trajectory \(1\).  

    Consider a particular day \(t\) that corresponds to trajectory 1.  We claim that eventually, if the algorithm runs longs enough, there will be at most \(k - 1\) active threads with higher rates than thread \(\prev^{(1)}(t)\).  

    We proceed iteratively.  If there are \(\le k - 1\) active threads with higher rates than thread \(\prev^{(1)}(t)\), we are done.  Otherwise, of the \(k\) threads with rate \(1, \dots, \frac{1}{k^2 \ln^2 k}\), there must be two of them that belong to the same trajectory \(q\).  Let \(\mathcal{T}^{(q)}_t\) be the set of days between \(\prev^{(1)}(t)\) and \(t\) that are associated with trajectory \(q\).  We can bound the total distance between any two solutions in \(\mathcal{T}^{(q)}_t\) by the total distance the \(q\)th trajectory moves in this time as
    \[\le \sum_{t' \in \mathcal{T}^{(q)}_t \setminus \{\prev^{(q)}(t)\}} \md \left(t', \mnext^{(q)}(t')  \right).\]

    Let \(t^{(q) *}\) be the thread of \(\mathcal{T}^{(q)}_t\) with the highest rate.  By the selection of \(q\), we know that \(t^{(q) *}\) has rate at least \(\frac{1}{(k - 1)^2 \ln^2 (k - 1)}\).  By \Cref{lem:quadratic-subsumed-time-bound}, this means that all other members of \(\mathcal{T}^{(q)}_t\) will be subsumed in  
    \[ \le k^3 \ln^2 k \cdot \sum_{t' \in \mathcal{T}^{(q)}_t \setminus \{\prev^{(q)}(t)\}} \md \left(t', \mnext^{(q)}(t')  \right)\]
    additional units of virtual radius.  Thus, after this many units of virtual radius, there can be at most one thread belonging to trajectory \(q\) that has a higher rate than \(\prev^{(1)}(t)\).  We say that trajectory \(q\) has ``collapsed" between \(\prev^{(1)}(t)\) and \(t\).

    Iterating this argument at most \(k - 1\) times gives us that after at most 
    \[k^3 \ln^2 k \cdot \sum_{q = 2}^k \sum_{t' \in \mathcal{T}^{(q)}_t \setminus \{\prev^{(q)}(t)\}} \md \left(t', \mnext^{(q)}(t')  \right)\]
    iterations of the algorithm, there can be at most \(k - 1\) active threads with higher rate than thread \(\prev^{(1)}(t)\).  Thus, thread \(\prev^{(1)}(t)\) has rate \(\ge \frac{1}{k^2 \ln^2 k}\).  

    This means that in at most \(k^2 \ln^2 k \cdot \md(S_{\prev^{(1)}(t)}, S_t)\) additional units of virtual radius, we have that 
    \[\md(S_{\prev^{(1)}(t)}, S_t) \le \radius(\prev^{(1)}(t)).\]
    Therefore, by \Cref{lem:subsume-completion}, there must be some thread that solves day \(t\) in at most 
    \[k^3 \ln^2 k \cdot \sum_{q = 2}^k \sum_{t' \in \mathcal{T}^{(q)}_t \setminus \{\prev^{(q)}(t)\}} \md \left(t', \mnext^{(q)}(t')  \right) ~+~ k^2 \ln^2 k \cdot \md(S_{\prev^{(1)}(t)}, S_t)\]
    units of virtual radius, and by \Cref{lem:quadratic-virtual-radius-bounds-total-radius} the virtual radius of an algorithm is within a \(O(1)\) factor of the total radius searched. 

    Now, we can sum the total radius searched over all days \(t\) belonging to trajectory 1.  
    The total radius searched over these days can be bounded by 
    \begin{align*}
        &\le \sum_{t \in \mathcal{T}^{(1)}} \left[ O(k^3 \ln^2 k) \cdot \sum_{q = 2}^k \sum_{t' \in \mathcal{T}^{(q)}_t \setminus \{\prev^{(q)}(t)\}} \md \left(t', \mnext^{(q)}(t')  \right) + O(k^2 \ln^2 k) \cdot \md(S_{\prev^{(1)}(t)}, S_t)\right] \\
        &\le O(k^3 \ln^2 k) \sum_{q = 1}^k \sum_{t \in \mathcal{T}^{(q)}} \md(S_{\prev^{(q)} (t)}, S_t) \\
        &\le O(k^3 \ln^2 k) \cdot \sum_{q = 1}^k \cost(\{P_t : t \in \mathcal{T}^{(q)}\}), 
    \end{align*}
    where the last line follows by \Cref{cor:predict-yesterday-is-competitive} because \(\sum_{t \in \mathcal{T}^{(q)}} \md(S_{\prev^{(q)} (t)}, S_t)\) is the cost of ``Predict yesterday's solution" run on trajectory \(q\) alone.  

    Finally, the same bound holds for all \(k\) trajectories, so we can conclude that the total radius searched by \Cref{alg:quadratic-decay} over all days is 
    \[\le O(k^4 \ln^2 k) \cdot \sum_{q = 1}^k \cost(\{P_t : t \in \mathcal{T}^{(q)}\}),\]
    and therefore \Cref{alg:quadratic-decay} is \(O(k^4 \ln^2 k)\) competitive against the \(k\) trajectories.  Finally, \Cref{lem:quadratic-total-radius-bounds-runtime} tells us that the total runtime of \Cref{alg:quadratic-decay} is at most \(O(1)\) times the total radius searched by \Cref{alg:quadratic-decay}.
\end{proof}

\begin{remark}
    Instead of running the threads at quadratically decaying rates, i.e., the \(i\)th fastest thread runs at rate \(\frac{1}{i^2 \ln^2 i}\), we could run them at harmonically decaying rates, i.e., the \(i\)th fastest thread runs at rate \(\frac{1}{i \ln^2 i}\).  This would result in an improved competitive ratio of \(O(k^3 \ln^2 k)\) against any \(k\) trajectories, in the sum of the radii searched (number of steps taken of the subroutine \(\alg\)).  However, the total runtime could scale as badly as an \(O(T)\) factor times the sum of the radii, where \(T\) is the time horizon.  

    The competitive ratio of the harmonic rate decay strategy is still worse than that of the reduction to \(k\)-server, which achieves competitive ratio \(O(k^2)\) (\Cref{cor:online-ball-search-via-k-server}).  However, the rate decay strategy also has the benefit that it is oblivious to the setting of \(k\), and therefore achieves the guarantee for all \(k\) simultaneously.  
\end{remark}

\section{Future Directions}


Our work leaves open a number of interesting research directions.  
\begin{itemize}
    \item Investigate whether it is possible to get more efficient algorithms for settings with a specific structure.  In this work, we considered the setting where we can group instances using a \(k\)-wise partition.  Are there other forms of structure that we can take advantage of? 
    \item Provide an algorithm with an improved competitive ratio for the online setting, or provide a lower bound for this setting.  Currently, the best lower bound we have is \(\Omega(k)\), which carries over from the offline setting (\Cref{rem:lower-bound}).  
    It would also be interesting to match the competitive ratio \(O(k^2)\) of the \(k\)-server based algorithm, but achieve it for all \(k\) simultaneously.  
    \item Find other warm-start and/or local-search type settings to which our techniques are applicable.
    \item In this work, we define the Online Ball Search problem (\Cref{def:online-ball-search}), and investigate some connections to the \(k\)-server problem.  Can algorithms for Online Ball Search be transformed into algorithms for \(k\)-server?  Alternatively, can we show a separation between these two problems?
\end{itemize}




\bibliographystyle{alpha}
\bibliography{ref}

\end{document}